\documentclass[aip,cha,amsmath,amssymb,reprint]{revtex4-1}
\usepackage{color}
\usepackage{graphicx}
\usepackage{amsmath}
\usepackage{amsthm}
\usepackage{bm}
\newtheorem{thm}{\protect\theoremname}
\theoremstyle{definition}

\theoremstyle{plain}
\newtheorem{prop}[thm]{Proposition}
\theoremstyle{plain}
\newtheorem{cor}[thm]{Corrolary}
\theoremstyle{plain}
\newtheorem{lem}[thm]{Lemma}
\theoremstyle{remark}

\draft 
\begin{document}


\title{Hierarchical frequency clusters in adaptive networks of phase oscillators} 



\author{Rico Berner}
\email[]{rico.berner@physik.tu-berlin.de}
\affiliation{Institute of Theoretical Physics, Technische Universit\"at Berlin, Hardenbergstr. 36, D-10623 Berlin, Germany}
\affiliation{Institute of Mathematics, Technische Universit\"at Berlin, Strasse des 17. Juni 136, D-10623 Berlin, Germany.}
\author{Jan Fialkowski}
\affiliation{Institute of Theoretical Physics, Technische Universit\"at Berlin, Hardenbergstr. 36, D-10623 Berlin, Germany}
\author{Dmitry Kasatkin}
\affiliation{Institute of Applied Physics of RAS, 46 Ul'yanov Street, 603950 Nizhny Novgorod, Russia}
\author{Vladimir Nekorkin}
\affiliation{Institute of Applied Physics of RAS, 46 Ul'yanov Street, 603950 Nizhny Novgorod, Russia}
\author{Serhiy Yanchuk}
\affiliation{Institute of Mathematics, Technische Universit\"at Berlin, Strasse des 17. Juni 136, D-10623 Berlin, Germany.}
\author{Eckehard Sch\"oll}
\affiliation{Institute of Theoretical Physics, Technische Universit\"at Berlin, Hardenbergstr. 36, D-10623 Berlin, Germany}


\date{\today}

\begin{abstract}
Adaptive dynamical networks appear in various real-word systems. One of the simplest phenomenological models for investigating basic properties of adaptive networks is the system of coupled phase oscillators with adaptive couplings. In this paper, we investigate the dynamics of this system. We extend recent results on the appearance of hierarchical frequency-multi-clusters by investigating the effect of the time-scale separation. We show that the slow adaptation in comparison with the fast phase dynamics is necessary for the emergence of the multi-clusters and their stability. Additionally, we study the role of double antipodal clusters, which appear to be unstable for all considered parameter values. We show that such states can be observed for a relatively long time, i.e., they are metastable. A geometrical explanation for such an effect is based on the emergence of a heteroclinic orbit.	
\end{abstract}


\maketitle 

\begin{quotation}
	Adaptive networks are characterized by the property that their
	connectivity can change in time, depending on the state of the network.
	A prominent example of adaptive networks are neuronal networks with
	plasticity, \textit{i.e.}, an adaptation of the synaptic coupling. Such an adaptation
	is believed to be related to the learning and memory mechanisms. In
	other real-world systems, the adaptivity plays important role as well~\cite{GRO08a}. This paper investigates a phenomenological model of
	adaptively coupled phase oscillators. The considered model is a natural
	extension of the Kuramoto system to the case with dynamical couplings.
	In particular, we review and provide new details on the self-organized
	emergence of multiple frequency-clusters.
\end{quotation}
\section{Introduction}\label{sec:intro}

One of the main motivations for studying adaptive dynamical networks
comes from the field of neuroscience where the weights of the synaptic
coupling can adapt depending on the activity of the neurons that are
involved in the coupling~\cite{GER96,CLO10,GER14a}.
For instance, the coupling weights can change in response to the relative
timings of neuronal spiking~\cite{MAR97a,ABB00,CAP08a,LUE16}. Adaptive
networks appear also in chemical~\cite{JAI01}, biological, or social
systems~\cite{GRO08a}.

This paper is devoted to a simple phenomenological model of adaptively
coupled phase oscillators. The model has been extensively studied
recently~\cite{KAS17,AOK09,AOK11,KAS16a,NEK16,GUS15a,PIC11a,TIM14,REN07,AVA18,BER19},
and it exhibits diverse complex dynamical behaviour. In particular,
stable multi-frequency clusters emerge in this system, when the oscillators
split into groups of strongly coupled oscillators with the
same average frequency. Such a phenomenon does not occur in the classical
Kuramoto or Kuramoto-Sakaguchi system. The clusters are shown to possess
a hierarchical structure, \textit{i.e.}, their sizes are significantly different~\cite{BER19}. Such a structure leads to significantly different frequencies
of the clusters and, as a result, to their uncoupling. This
phenomenon is reported for an adaptive network of Morris-Lecar bursting
neurons with spike-timing-dependent plasticity rule~\cite{POP15}. In addition, the role of hierarchy and modularity in brain networks has been discussed recently~\cite{Bassett2008,Bassett2010,Lohse2014,Betzel2017,Ashourvan2019}. Both features therefore seem to play a key role for real-world neural networks. Along these lines, we study a dynamical model to analyse the self-organized formation of such network structures.

In this paper, we first provide a non-technical overview of known
results on multi-frequency clusters from Refs.~\onlinecite{KAS17,BER19}.
Apart from that, we investigate the role of the time-scale separation.
Particularly, we find that the slow adaptation mechanism
in comparison with the fast dynamics of the oscillators is an important
necessary ingredient for the emergence of stable multi-clusters. Discussing
the stability, we point out that the stability analysis for the limiting
case without adaptation does not provide correct stability results
for arbitrarily small adaptation. In addition, the stability for one-clusters is described depending on the time-scale separation. With these results, we quantitatively relate the stability of one- with multi-frequency cluster states which goes beyond the qualitative analysis given in Ref.~\onlinecite{BER19}. Providing a novel result for the stability of two-clusters, the instability of evenly sized clusters is shown. Finally, we discuss the role of a special
type of cluster states, called \textit{double antipodal} clusters. We show that these states are unstable for all parameter values but can appear
as saddles connecting synchronous and splay states in phase space. As a result, such
states can be observed as a "meta-stable" transition between the
phase-synchronous and non-phase-synchronous state. Moreover, the double antipodal states are shown to play an important role for the global dynamics of the adaptive system.

The structure of the paper is as follows. Section~\ref{sec:model}
presents the model; Sec.~\ref{sec:Hcluster} and~\ref{sec:blocks} provide a non-technical
overview of Ref.~\onlinecite{BER19} on the multi-frequency cluster states and their classification. The new results are included in Sec.~\ref{sec:1Clstability}-\ref{sec:MClstability}. In
Sec.~\ref{sec:blocks} we further describe one-clusters, i.e., single
blocks of which the multi-clusters are composed. Section~\ref{sec:1Clstability} discusses the results of Ref.~\onlinecite{BER19} on the stability of one- and multi-frequency clusters from a different viewpoint focusing on the influence of the time-scale separation. We provide novel rigorous results on the stability for the whole classes of antipodal, double antipodal, and $4$-phase cluster states. The stability regions for antipodal and splay-type one-clusters are explicitly described for any parameter range of the time-scale separation. In Sec.~\ref{sec:doubleAntipodal}, the role of double antipodal states for the global dynamics of the system is discussed. The construction of multi-cluster states from one-clusters is demonstrated in Sec.~\ref{sec:RoleSlowAdap}. For the existence of multi-clusters a new upper bound for the time separation is derived. In Sec.~\ref{sec:MClstability}, we connect the stability properties of one- and multi-clusters. As conjectured in Ref.~\onlinecite{BER19}, the instability of evenly sized two-clusters of splay type are proved. For the sake of readability, the proofs for any of the statements in this section are provided in the appendix~\ref{sec:app_StabOneCluster}. We end with the Conclusion.

\section{Model}\label{sec:model}
In this article, we consider a network of $N$
adaptively coupled phase oscillators 
\begin{align}
	\frac{d\phi_{i}}{dt} & =\omega-\frac{1}{N}\sum_{j=1}^{N}\kappa_{ij}\sin(\phi_{i}-\phi_{j}+\alpha),\label{eq:PhiDGL_general}\\
	\frac{d\kappa_{ij}}{dt} & =-\epsilon\left(\sin(\phi_{i}-\phi_{j}+\beta)+\kappa_{ij}\right),\label{eq:KappaDGL_general}
\end{align}
where $\phi_{i}\in[0,2\pi)$ is the phase of the $i$-th oscillator
($i=1,\dots,N$) and $\omega$ the natural frequency. The oscillators
interact according to the coupling structure represented by the coupling
weights $\kappa_{ij}$ ($i,j=1,\dots,N$) as dynamical variables.
The parameter $\alpha$ can be considered as a phase-lag of the interaction~\cite{SAK86}.
This paradigmatic model of an adaptive network Eqs.~(\ref{eq:PhiDGL_general})\textendash (\ref{eq:KappaDGL_general})
has attracted a lot of attention recently \cite{KAS17,AOK09,AOK11,KAS16a,NEK16,GUS15a,PIC11a,TIM14,REN07,AVA18,BER19}.
It provides a generalization of the Kuramoto-Sakaguchi model with fixed $\kappa_{ij}$~\cite{ACE05a,KUR84,OME12b,STR00,PIK08}.

The coupling topology of the network at time $t$ is characterized
by the coupling weights $\kappa_{ij}(t)$. With a small parameter
$0<\epsilon\ll1$, the dynamical equation (\ref{eq:KappaDGL_general})
describes the adaptation of the network topology depending on the
dynamics of the network nodes. In the neuroscience
context, such an adaptation can be also called plasticity~\cite{AOK11}.
The chosen adaptation function has the form $-\sin(\phi_{i}-\phi_{j}+\beta)$
with control parameter $\beta$. With this, different plasticity rules
can be modelled, see Fig.~\ref{fig:Plasticity_fucntion_betadep}.
For instance, for $\beta=-\pi/2$, a Hebbian-like rule is obtained
where the coupling $\kappa_{ij}$ is increasing between any two phase
oscillators with close phases, i.e., $\phi_{i}-\phi_{j}$ close to
zero~\cite{HEB49,HOP96,SEL02,AOK15} (fire together - wire together). If $\beta=0$, the link $\kappa_{ij}$
is strengthened if the $j$-th oscillator precedes the $i$-th.
Such a relationship promotes a causal structure in the oscillatory
system. In neuroscience these adaptation rules are typical for spike timing- dependent
plasticity~\cite{CAP08a,MAI07,LUE16,POP13}. 
\begin{figure}
	\includegraphics{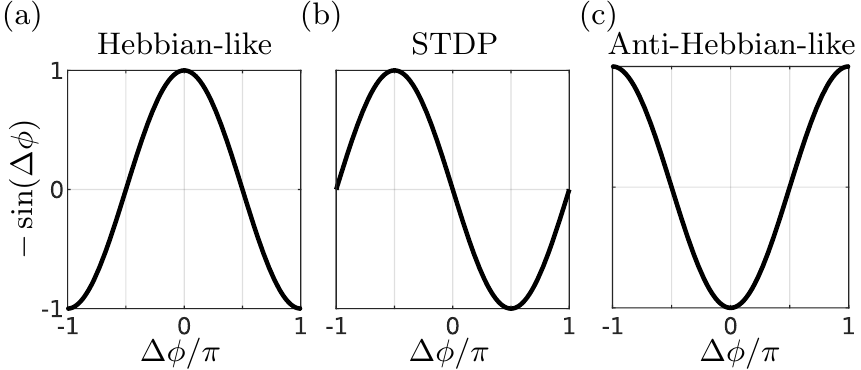}
	\caption{\label{fig:Plasticity_fucntion_betadep} 
		The plasticity function $-\sin(\Delta\phi+\beta)$ and corresponding plasticity rules are presented. (a)~$\beta=-\frac{\pi}{2}$ (Hebbian), (b)~$\beta=0$ (Spike timing-dependent plasticity, STDP), (c) $\beta=\frac{\pi}{2}$ (Anti-Hebbian).}
\end{figure}

Let us mention important properties of the model~\eqref{eq:PhiDGL_general}\textendash \eqref{eq:KappaDGL_general}.
The parameter $\epsilon\ll1$ separates the time scales of the slowly
adapting coupling weights from the fast moving phase oscillators.
Further, the coupling weights are confined to the interval $-1\le\kappa_{ij}\le1$
due to the  fact that $d\kappa_{ij}/dt\le0$ for $\kappa_{ij}=1$ and $d\kappa_{ij}/dt\ge0$
for $\kappa_{ij}=-1$, see Ref.~\onlinecite{KAS17}. Due to the invariance of
system (\ref{eq:PhiDGL_general})\textendash (\ref{eq:KappaDGL_general})
with respect to the shift $\phi_{i}\mapsto\phi_{i}+\psi$ for all
$i=1,\dots,N$ and $\psi\in[0,2\pi)$, the frequency $\omega$ can
be set to zero in the co-rotating coordinate frame $\phi\mapsto\phi+\omega t$.
Finally we mention the symmetries of the system~\eqref{eq:PhiDGL_general}\textendash \eqref{eq:KappaDGL_general}
with respect to the parameters $\alpha$ and $\beta$: 
\begin{align*}
	(\alpha,\beta,\phi_{i},\kappa_{ij}) & \mapsto(-\alpha,\pi-\beta,-\phi_{i},\kappa_{ij}),\\
	(\alpha,\beta,\phi_{i},\kappa_{ij}) & \mapsto(\alpha+\pi,\beta+\pi,\phi_{i},-\kappa_{ij}).
\end{align*}
These symmetries allow for a restriction of the analysis to the parameter region
$\alpha\in[0,\pi/2)$ and $\beta\in[-\pi,\pi)$.

The Kuramoto order parameter $R_{1}$  measures the synchrony of phase oscillators
$\bm{\phi}=(\phi_{1},\dots,\phi_{N})^{T}$. Correspondingly, the $l$-th
moment order parameter $R_{l}$ is given by 
\begin{align}
	R_{l}(\bm{\phi}):=\left|\frac{1}{N}\sum_{j=1}^{N}e^{\mathrm{i}l\phi_{j}}\right|
\end{align}
where $l\in\mathbb{N}$. In the following section we will use this
quantity in order to characterize several dynamical states. 
\section{Hierarchical frequency clusters}\label{sec:Hcluster}
System (\ref{eq:PhiDGL_general})--(\ref{eq:KappaDGL_general})
has been studied numerically in Refs.~\onlinecite{AOK09,AOK11,NEK16,KAS16a,KAS17}.
In particular, it is shown that starting from uniformly distributed
random initial condition ($\phi_{i}\in[0,2\pi)$, $\kappa_{ij}\in[-1,1]$
for all $i,j\in\left\{ 1,\dots,N\right\} $) the system can reach
different multi-frequency cluster states with hierarchical structure.
An individual cluster in the multi-cluster state consists of frequency
synchronized groups of phase oscillators. In the following we discuss
the structural form of these clusters. Subsequently, frequency multi-clusters states are called multi-cluster states or multi-clusters for simplicity.

Figure~\ref{fig:MC_state_general} shows a hierarchical multi-cluster
state. The solution was obtained by integrating the system (\ref{eq:PhiDGL_general})--(\ref{eq:KappaDGL_general})
numerically starting from uniformly distributed random initial conditions.
The self-couplings $\kappa_{ii}$ are set to zero in numerical simulations, since they
do not influence the  relative dynamics of the system 
\footnote{All variables $\kappa_{ii}$ converge to $-\sin\beta$, which leads to the same constant term $\sin(\beta) \sin(\alpha) / N$ in the right hand side of Eq.~(\ref{eq:PhiDGL_general}). The latter term can be absorbed in the co-rotating coordinate frame}. 
We re-order the oscillators (after sufficiently long transient time)
by first sorting the oscillators with respect to their average
frequencies. After that the oscillators with the same frequency are
sorted by their phases. Figure~\ref{fig:MC_state_general}(a) displays
the coupling matrix (right) of the multi-cluster state and a representation
of the coupling structure as a network graph (left). The coupling
matrix demonstrates a clear splitting into three groups. This splitting
is also visible in the graph representation of the coupling network.
The coupling weights between oscillators of the same group vary in
a larger range than between those of different groups which are generally smaller in magnitude. The splitting
into three groups is manifested in the behaviour of the phase oscillators,
as well. We find that the oscillators of the same group posses the
same constant frequency with possible phase lags, Fig.\ref{fig:MC_state_general}(b).
We call the groups of oscillators (frequency) clusters and the
corresponding dynamical states multi-(frequency)cluster states. 
\begin{figure}
	\includegraphics{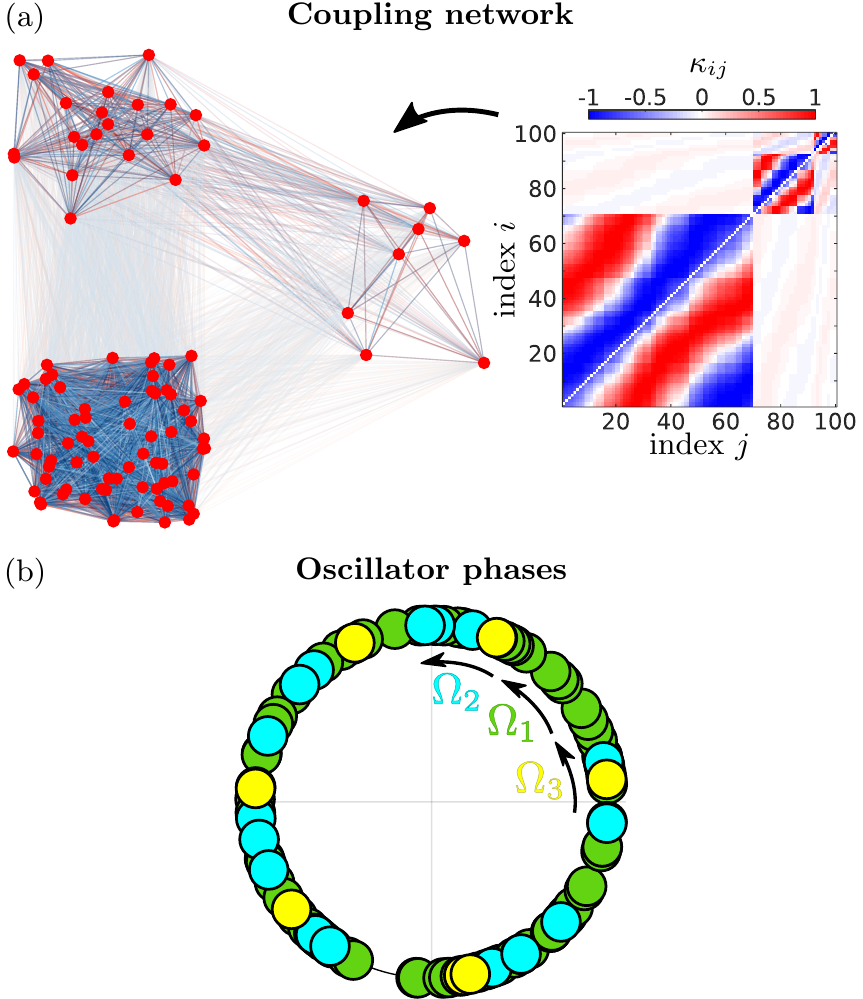}
	\caption{\label{fig:MC_state_general} Three-frequency cluster of splay type
		at $t=10000$. (a) Coupling weights represented as a graph (left)
		and as a coupling matrix (right). In the graph representation, the
		dynamical vertices are represented by red nodes and the edges are
		coloured with respect to the coupling weight. Red and blue refer to
		positive and negative coupling weights, respectively. Light and dark
		colours refer to weak and strong coupling weights, respectively. (b)
		Distribution of the phases $\phi_{i}$ for each of the three clusters.
		Each node represents one oscillator and is coloured with respect to
		the cluster to which it belongs. Parameter values: $\epsilon=0.01$,
		$\alpha=0.3\pi$, $\beta=0.23\pi$, $\omega=0$, and $N=100$.}
\end{figure}

In a multi-cluster, the coupling matrix $\kappa$ can be
divided into different blocks; $\kappa_{ij,\mu\nu}$ will refer to the
coupling weight between the $i$-th oscillator of the $\mu$-th cluster
to the $j$-th oscillator of the $\nu$-th cluster. Analogously, $\phi_{i,\mu}$
denotes the $i$-th phase oscillator in the $\mu$-th cluster. In
general, the temporal behaviour for each oscillator in a $M$-cluster
state takes the form 
\begin{align}
	\phi_{i,\mu}(t)=\Omega_{\mu}t+a_{i,\mu}+s_{i,\mu}(t) &  & \begin{split}\mu=1,\dots,M\\
		i=1,\dots,N_{\mu}
	\end{split}
	\label{eq:MCstate_gen}
\end{align}
where $M$ is the number of clusters, $N_{\mu}$ is the
number of oscillators in the $\mu$-th cluster, $a_{i,\mu}\in[0,2\pi)$
are phase lags, and $\Omega_{\mu}\in\mathbb{R}$ is the collective
frequency of the oscillators in the $\mu$-th cluster. 
The functions $s_{i,\mu}$ are bounded.

The numerical analysis of system \eqref{eq:PhiDGL_general}--\eqref{eq:KappaDGL_general}
shows the appearance of different multi-cluster states depending
on  particular choices of the phase lag parameters $\alpha$ and $\beta$
as well as on initial conditions. Starting from random initial conditions
the system can end up in several states such as multi-clusters and
chimera-like states~\cite{KAS17}. Figure~\ref{fig:MC_states_types}
shows examples for the three types of multi-cluster states which appear
dynamically in \eqref{eq:PhiDGL_general}--\eqref{eq:KappaDGL_general}.

\subsection{Splay type multi-clusters}

The first type is called splay type multi-cluster state, see Fig.~\ref{fig:MC_states_types}(a).
The separation into three clusters is clearly visible in the coupling
matrix, as well as a hierarchical structure in the cluster sizes. Regarding
the distribution of the phases, we notice that the oscillators from
each group are almost homogeneously dispersed on the circle. In fact,
the phases from each cluster  fulfil  the condition $R_{2}(\bm{\phi}_{\mu})=0$
($\mu=1,2,3$). Note that splay states as they are defined in several
other works~\cite{NIC92,STR93a,CHO09} share the property $R_{1}(\bm{\phi})=0$.
This property can be seen as a measure of incoherence for the oscillator
phases, as well. In fact, it was shown that splay states are part
of a whole family of solutions~\cite{BUR11,ASH16a} given by exactly
$R_{1}(\bm{\phi})=0$. Further, $R_{2}(\phi)=R_{1}(2\bm{\phi})$ relates
the two measures of incoherence. These facts motivate the definition of those clusters with $R_{2}(\bm{\phi})=0$ as \emph{splay-type clusters}.

The temporal behaviour for all phase oscillators in the splay multi-cluster
state is characterized by a constant frequency which differs for
the different clusters, \textit{i.e.}, according to \eqref{eq:MCstate_gen},
$\phi_{i,\mu}(t)=\Omega_{\mu}t+a_{i,\mu}$ with $R_{2}(\bm{a}_{\mu})=0$
for all $\mu=1,2,3$ and $i=1,\dots,N_{\mu}$. In addition, the hierarchical
cluster sizes are reflected in the frequencies. Oscillators of a big cluster
have a higher frequency than those of smaller clusters. The coupling
weights between the phase oscillators are fixed or change periodically
with time depending on whether the oscillators belong to the same
or different clusters, respectively. Moreover, the amplitude of coupling
weights between clusters depends on the frequency difference of the
corresponding clusters. The higher the frequency difference, the
smaller is the amplitude. The periodic behaviour of the coupling
weights between clusters is present in all types of multi-cluster
states (Fig.~\ref{fig:MC_states_types}(a,b,c)).

\subsection{Antipodal type multi-clusters}

Figure~\ref{fig:MC_states_types}(b) shows another possible multi-cluster
state. As in Fig.~\ref{fig:MC_states_types}(a) the clusters are
clearly visible and their oscillators show frequency synchronized temporal
behaviour. In addition, the time series for the oscillators show periodic
modulations on top of the linear growth. This additional dynamics
is the same for all oscillators of the same cluster, and hence they
are still temporally synchronized. We have $\phi_{i,\mu}(t)=\Omega_{\mu}t+a_{i,\mu}+s_{\mu}(t)$.
In analogy to the coupling weights between the clusters, the amplitudes
of the bounded function $s_{\mu}(t)$ depend on the differences of
the cluster frequencies.

In contrast to the splay states, the phase distribution fulfils $R_{2}(\bm{a}_{\mu})=1$
for all $\mu=1,2,3$, see Fig~\ref{fig:MC_states_types}(b), middle
panel. Hence, all oscillators of a cluster have either the same phase
$a_{\mu}\in[0,2\pi)$ or the antipodal phase $a_{\mu}+\pi$ such that
$2a_{i,\mu}=2a_{\mu}$ modulo $2\pi$ for all $i=1,\dots,N_{\mu}$. Therefore, the
clusters represented in Fig~\ref{fig:MC_states_types}(b) are called
\emph{antipodal type multi-cluster}. Note that with this formal definition of an antipodal state, in-phase clusters belong to the class of antipodal clusters.

\subsection{Mixed type multi-clusters}

The third type of multi-cluster states combines the previous two types.
The 2-cluster state shown in Fig.~\ref{fig:MC_states_types}(c) consists
of one splay cluster and one antipodal cluster. We call these states
\emph{mixed type multi-cluster states}. As we have seen before, the interaction of a cluster with
an antipodal cluster induces a modulation $s(t)$ additional to the
linear growth of the oscillator's phase. In contrast, the interaction
with a splay cluster does not introduce any modulation. Thus, the temporal
dynamics of the oscillators in the antipodal cluster ($\mu=1$) have
$s_{i,1}(t)\equiv0$  while the  oscillators in the splay cluster
($\mu=2$) show additional bounded modulations $s_{i,2}(t)$, see Fig.~\ref{fig:MC_states_types}(c). For the oscillators of the splay cluster we plot the time series of two representatives.
We notice the temporal shift in the dynamics of the  two representatives
of the splay cluster. The oscillators in the splay cluster are not
completely temporally synchronized. More specifically we have $\phi_{i,1}(t)=\Omega_{1}t+a_{i,1}$
with $R_{2}(\bm{a}_{1})=1$ for $i=1,\dots,N_{1}$ and $\phi_{i,2}(t)=\Omega_{2}t+a_{i,2}+s_{i,2}(t)$
with $R_{2}(\bm{a}_{2})=0$ for $i=1,\dots,N_{2}$. 
\begin{figure*}
	\centering \includegraphics{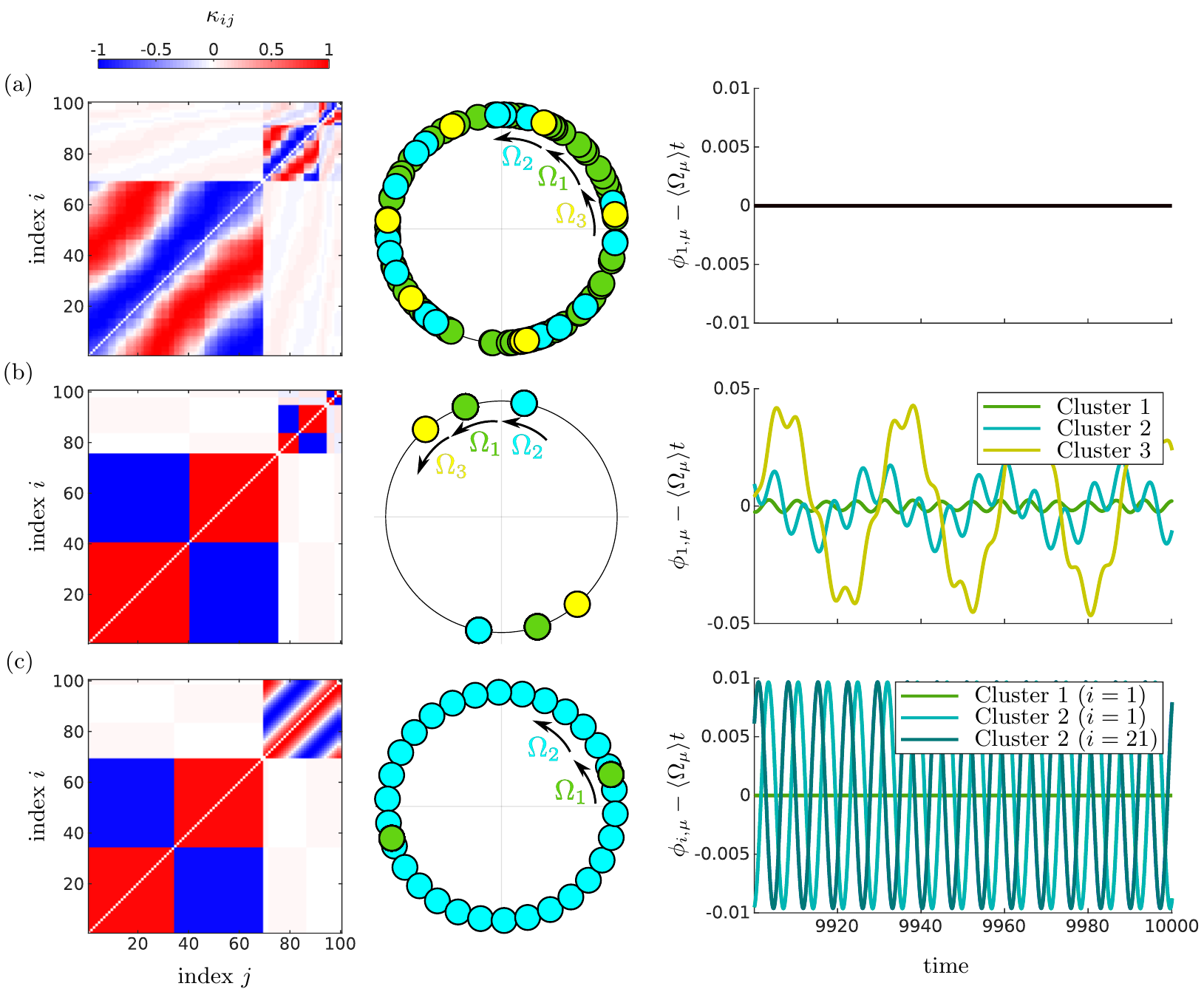}
	\caption{\label{fig:MC_states_types} Three different types of multi-cluster
		states at $t=10000$ with $N=100$ and $\epsilon=0.01$. For all types,
		the coupling matrix (left), distribution of the phases (middle), and
		time series of representative phase oscillators from each cluster
		(right) are presented. In the plot of the phase distribution, each
		node represents one oscillator and is colored with respect to the
		cluster to which it belongs. The time series are shown after subtracting
		the average linear growth $\phi_{i,\mu}(t)-\langle\Omega_{\mu}\rangle t$.
		The colouring of the time series (shaded for visibility) of a representative phase oscillator from one cluster is in accordance with the pictures in the middle panel. (a) Splay type 3-cluster for $\alpha=0.3\pi$,
		$\beta=0.23\pi$; (b) Antipodal type 3-cluster for $\alpha=0.3\pi$,
		$\beta=-0.53\pi$; (c) Mixed type 2-cluster for $\alpha=0.3\pi$,
		$\beta=-0.4\pi$.}
\end{figure*}

Despite the complexity of the three types of multi-clusters states,
the structures can be broken down into simple blocks. In fact, one-cluster
states of splay and antipodal type serve as building blocks in order
to create more complex multi-cluster structures. In the following
section we will analyse these blocks. The building of higher cluster
structures will be discussed for 2-cluster states of splay type. 
\section{One-cluster states}\label{sec:blocks}
As we have seen in Section~\ref{sec:Hcluster},
certain one-cluster states serve as building blocks for higher multi-cluster
states. In this section, we review the basic properties of one-cluster states and conclude that their shape and existence is independent of the time-separation parameter, see Ref.~\onlinecite{BER19} for more details. 

Formally a one-cluster state is one group of frequency synchronized
phase oscillators given by
\begin{align*}
\phi_{i}=\Omega t+a_{i}
\end{align*} with
$a_{i}\in[0,2\pi)$ ($i=1,\dots,N$) and
constant coupling weights  
\begin{align}\label{eq:1Clcouplings}
	\kappa_{ij}=-\sin(a_{i}-a_{j}+\beta)
\end{align}
($i,j=1,\dots,N$).

Figure~\ref{fig:1CL_states_types} shows three possible types of
one-cluster states for the system \eqref{eq:PhiDGL_general}--\eqref{eq:KappaDGL_general}.
It has been shown that these are the only existing types of one-cluster
states~\cite{BER19}. The first two shown in Fig.~\ref{fig:1CL_states_types}(a,b)
are the splay ($R_{2}(\mathbf{a})=0$) and the antipodal ($R_{2}(\mathbf{a})=1$)
clusters which were already discussed in Sec.~\ref{sec:Hcluster}.
The third type in Fig.~\ref{fig:1CL_states_types}(c) consists of
two groups of antipodal phase oscillators with a fixed phase lag $\psi$.
We call this class of states \emph{double antipodal}. As it was mentioned in Sec.~\ref{sec:Hcluster}.B., the formal definition of an antipodal state includes full in-phase relation of the oscillators. Thus in extension of the typical configuration shown in Fig.~\ref{fig:1CL_states_types}(c), there exist double antipodal states where only three or even only two different phases are occupied. The constant $\psi$ is the unique (modulo $2\pi$) solution of the equation 
\begin{align}\label{eq:DoubleAntipodalCond}
	\frac{1-q}{q}\sin(\psi-\alpha-\beta)=\sin(\psi+\alpha+\beta),
\end{align}
where $q=Q/N$ and $Q$ is the number of phase shifts $a_{i}$ such that $a_{i}\in\{0,\pi\}$.
The corresponding frequencies for the three types of states are 
\begin{align}
	\Omega=\begin{cases}
		\frac{1}{2}\cos(\alpha-\beta) & \mbox{if }\,R_{2}(\mathbf{a})=0,\\
		\sin\alpha\sin\beta & \mbox{if }\,R_{2}(\mathbf{a})=1,\\
		\frac{1}{2}\left[\cos(\alpha-\beta)-R_{2}(\mathbf{a})\cos\psi\right] & \text{if double antipodal}
	\end{cases}\label{eq:OmegaPL}
\end{align}
\begin{figure}
	\centering \includegraphics{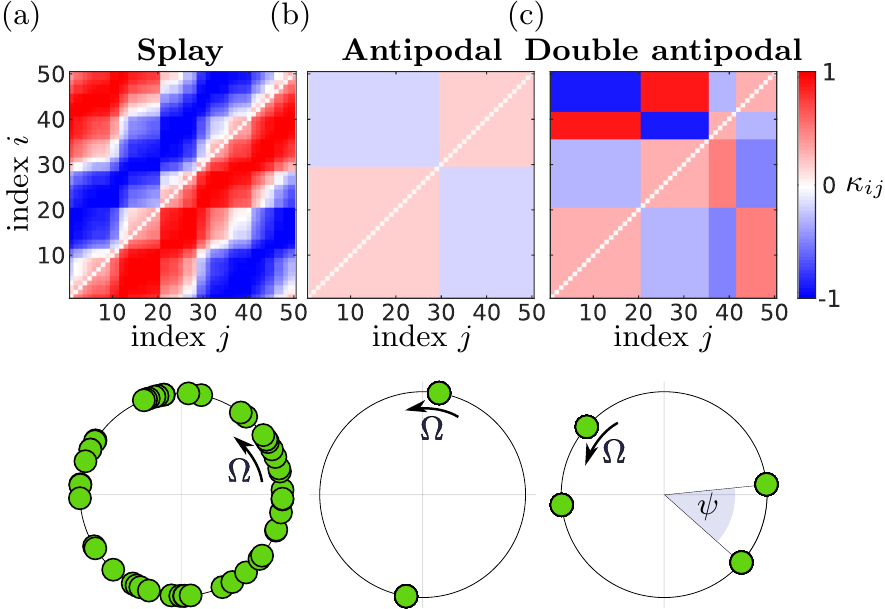}
	\caption{\label{fig:1CL_states_types} All possible types of one-cluster states
		for system Eq.~\eqref{eq:PhiDGL_general}--\eqref{eq:KappaDGL_general}.
		(a) Splay state; (b) antipodal state; (c) double antipodal state.}
\end{figure}

Note that the condition $R_{2}(\mathbf{a})=0$ gives rise to a $N-2$
dimensional family of solutions. Well-known representatives of this
family are rotating-wave states which have the following form $\mathbf{a}_{k}=(0,2\pi k/N,\dots,2\pi k(N-1)/N)^{T}$
for any wave-number $k=1,\dots,N-1$. The existence as well as the
explicit form for any of the three types of one-cluster states do
not depend on the time separation parameter $\epsilon$. As long as
$\epsilon>0$, the building blocks appear to be solutions of the system \eqref{eq:PhiDGL_general}--\eqref{eq:KappaDGL_general}.
\section{Stability of one-cluster states}\label{sec:1Clstability}
In sections~\ref{sec:Hcluster} and \ref{sec:blocks}
a large number of co-existing multi-cluster states were discussed.
As multi-clusters are constructed out of one-cluster states, studying the stability of the building blocks is of major importance. In contrast to the existence of one-cluster states, the stability of those states depend crucially on the time-separation parameter. We analyse this dependency below.

The diagram in Fig.~\ref{fig:StabRegions} shows the regions
of stability for antipodal and rotating-wave one-cluster states. The
diagram is based on analytic results which has been recently found. In appendix~\ref{sec:app_StabOneCluster}, we review Corollary~4.3 from Ref.~\onlinecite{BER19} as Proposition~\ref{prop:LinerizedOneCluster_RW_Spectrum} and provide an extension to the whole class of antipodal states, see Corollary~\ref{cor:StabAntipodal}. Further, all other proofs are provided in the appendix~\ref{sec:app_StabOneCluster}, as well.

A linear stability analysis for double antipodal states shows
that they are unstable for all parameters $\alpha$ and $\beta$, see Corollary~\ref{cor:StabDoubleAntipodal} in the appendix~\ref{sec:app_StabOneCluster}. The role of the double antipodal clusters are discussed further
in Sec.~\ref{sec:doubleAntipodal}. 

In Fig.~\ref{fig:StabRegions} the regions of stability are presented
for several values of the time separation parameter $\epsilon$. The
first case in panel (a) assumes $\epsilon=0$, where the network structure
is non-adaptive but fixed to the values given by the one-cluster states,
\textit{i.e.}, $\kappa_{ij}=-\sin(a_{i}-a_{j}+\beta)$ as given in
Sec.~\ref{sec:blocks}. The linearised system in this case is given
by 
\begin{multline}
\frac{\mathrm{d}\delta\phi_{i}}{\mathrm{d}t}=-\frac{1}{2N}\sum_{j=1}^{N}\left(\sin(\alpha-\beta)\right.\\
\left. -\sin(2(a_{i}-a_{j})+\alpha+\beta)\right)\left(\delta\phi_{i}-\delta\phi_{j}\right).\label{eq:linear}
\end{multline}
For the synchronized or antipodal state, the value $2a_{i}\,\text{mod }2\pi$
is the same for all $i$. Hence, the term $2(a_{i}-a_{j})$ disappears
and the linearised system (\ref{eq:linear}) possesses the same form
as the linearised system for the synchronized state of the Kuramoto-Sakaguchi
system~\cite{BUR11} with coupling constant $\sigma(\beta)=-\sin(\beta)$. As it follows from Ref.~\onlinecite{BUR11},
the synchronized as well as all other antipodal states are stable
for $\sigma(\beta)\cos(\alpha)>0$. The region of stability of the
rotating-wave cluster has a more complex shape, see hatched area in
Fig~\ref{fig:StabRegions}(a). We find large areas where both types
of one-cluster states are stable simultaneously, as well as the regions where
no frequency synchronized state is stable. The results shown in Fig~\ref{fig:StabRegions}(a) are
in agreement with Ref.~\onlinecite{AOK11}, where the authors consider
the case $\epsilon=0$ in order to approximate the limit case of extremely
slow adaptation $\epsilon\to0$. However, such an approach for studying
the stability of clusters for small adaptation is not correct in general.
As Figs.~\ref{fig:StabRegions}(b-d) show, the stability of the network
with small adaptation $\epsilon>0$ is different. 

The case $\epsilon=0.01$ is shown in Fig~\ref{fig:StabRegions}(b),
where we observe regions for stable antipodal and rotating-wave states
as well. The introduction of a small but non-vanishing adaptation
changes the regions of stability significantly. The diagram in Fig~\ref{fig:StabRegions}(b)
remains qualitatively the same for smaller values of $\epsilon$. This can be read of from the analytic findings presented in the Proposition~\ref{prop:LinerizedOneCluster_RW_Spectrum} and Corollary~\ref{cor:StabAntipodal}. The changes in the stability areas are due to subtle changes in the equation which determine the eigenvalue of the corresponding linearised system. In fact, the  adaptation introduces the necessary condition $\sin(\alpha+\beta)<0$ for the stability of antipodal states. Additionally, for all splay states, including rotating-wave states, the necessary condition $\sin(\alpha-\beta)+2\epsilon>0$ is introduced, see Proposition~\ref{prop:NecCondSplay}. This is why one can observe a non-trivial effect of adaptivity on the stability in Figs.~\ref{fig:StabRegions}. In particular, the parameter $\beta$, which determines the plasticity
rule, has now a non-trivial impact on the stability of antipodal states for any $\epsilon>0$.
As it can be seen in Fig.~\ref{fig:StabRegions}(b), one-cluster states
of antipodal type are supported by a Hebbian-like adaption ($\beta\approx-\pi/2$)
while splay states are supported by causal rules ($\beta\approx0$).
For the asynchronous region, the dynamical system~\eqref{eq:PhiDGL_general}\textendash \eqref{eq:KappaDGL_general}
can exhibit very complex dynamics and show chaotic motion~\cite{KAS16a}.
This region is supported by an anti-Hebbian-like rule ($\beta\approx\pi/2$).

By increasing the parameter $\epsilon$, see Fig.~\ref{fig:StabRegions}(c,d),
two observations can be made. First, the region of asynchronous dynamical
behaviour is shrinking. For $\epsilon=1$, we find at least one stable
one-cluster state for any choice of the phase lag parameters $\alpha$
and $\beta$. Secondly, the regions where both types of one-cluster
states are stable are shrinking as well. In the limit of instant network adaptation,
\textit{i.e.}, $\epsilon\to\infty$, the stability regions are completely
separated. Both types of one-cluster states divide the whole parameter
space into two areas. In this case, the boundaries are described by
$\alpha+\beta=0$ and $\alpha+\beta=\pi$. This division can be seen from the analytic findings presented in Proposition~\ref{prop:LinerizedOneCluster_RW_Spectrum} (Appendix A). In the case of antipodal states, the quadratic equation which determines the Lyapunov coefficients has negative roots if and only if $\epsilon+\cos(\alpha)\sin(\beta)>0$ and $\sin(\alpha+\beta)<0$. Here, even for $\epsilon> 1$, the condition $\sin(\alpha+\beta)<0$ is the only remaining one. Similarly, we find $\sin(\alpha+\beta)>0$ as a condition for the stability of rotating-wave states for $\epsilon\to\infty$.
\begin{figure}
	\centering \includegraphics{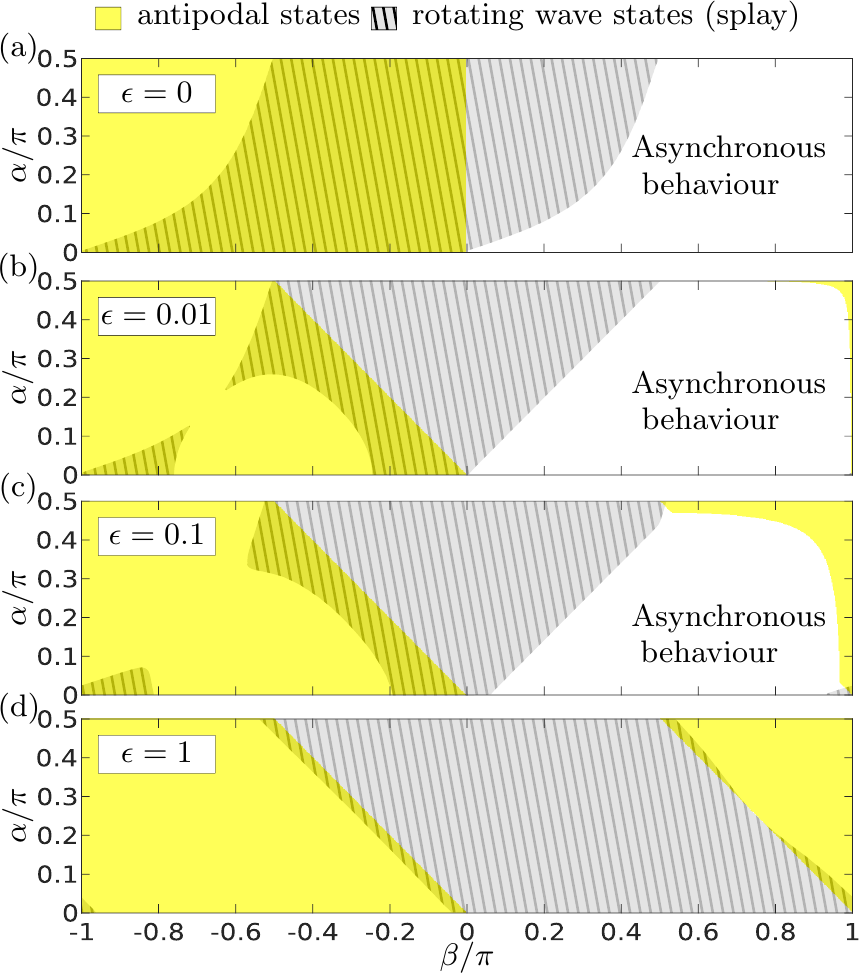}
	\caption{\label{fig:StabRegions} The regions of stability for antipodal and
		rotating-wave states are presented in ($\alpha,\beta$) parameter space
		for different values of $\epsilon$. Coloured and hatched areas refer
		to stable regions for these states as indicated in the legend. White areas
		refer to region where these one-cluster states are unstable. (a) $\epsilon=0$;
		(b) $\epsilon=0.01$; (c) $\epsilon=0.1$; (d) $\epsilon=1$.}
\end{figure}
\section{Double antipodal states}\label{sec:doubleAntipodal}
Splay and antipodal clusters serve as
building blocks for multi-cluster states. The third type, the double
antipodal clusters, are not of this nature since they appear to be
unstable everywhere. As unstable objects, they can still play an important
role for the dynamics. Here we would like to present an example, where
the double antipodal clusters become part of a simple heteroclinic network.
As a result, they can be observed as metastable states in numerics.

As an example, we first analyse the system of $N=3$ adaptively coupled
phase oscillators which is the smallest system with a double antipodal state. According to the definition of a double antipodal state, laid out in Sec.~\ref{sec:blocks}, the phases $a_i$ of the oscillators $\phi_i$ are allowed to take values from the set $\{0,\pi,\psi,\psi+\pi\}$ where $\psi$ uniquely  solves Eq.~(\ref{eq:DoubleAntipodalCond}). Further, at least one oscillator $\phi_i$ with $a_i\in\{0,\pi\}$ and one oscillator $\phi_j$ ($j\ne i$) with $a_j\in\{\psi,\psi+\pi\}$ are needed in order to represent one of the two antipodal groups. Note that for the parameters given in Fig.~\ref{fig:DoubleAntipodal} and $N=3$, the equation~(\ref{eq:DoubleAntipodalCond}) yields $\psi=1.602\pi$ if $a_1,a_2\in\{0,\pi\}$ and $a_3\in\{\psi,\psi+\pi\}$.

\begin{figure}
	\centering \includegraphics{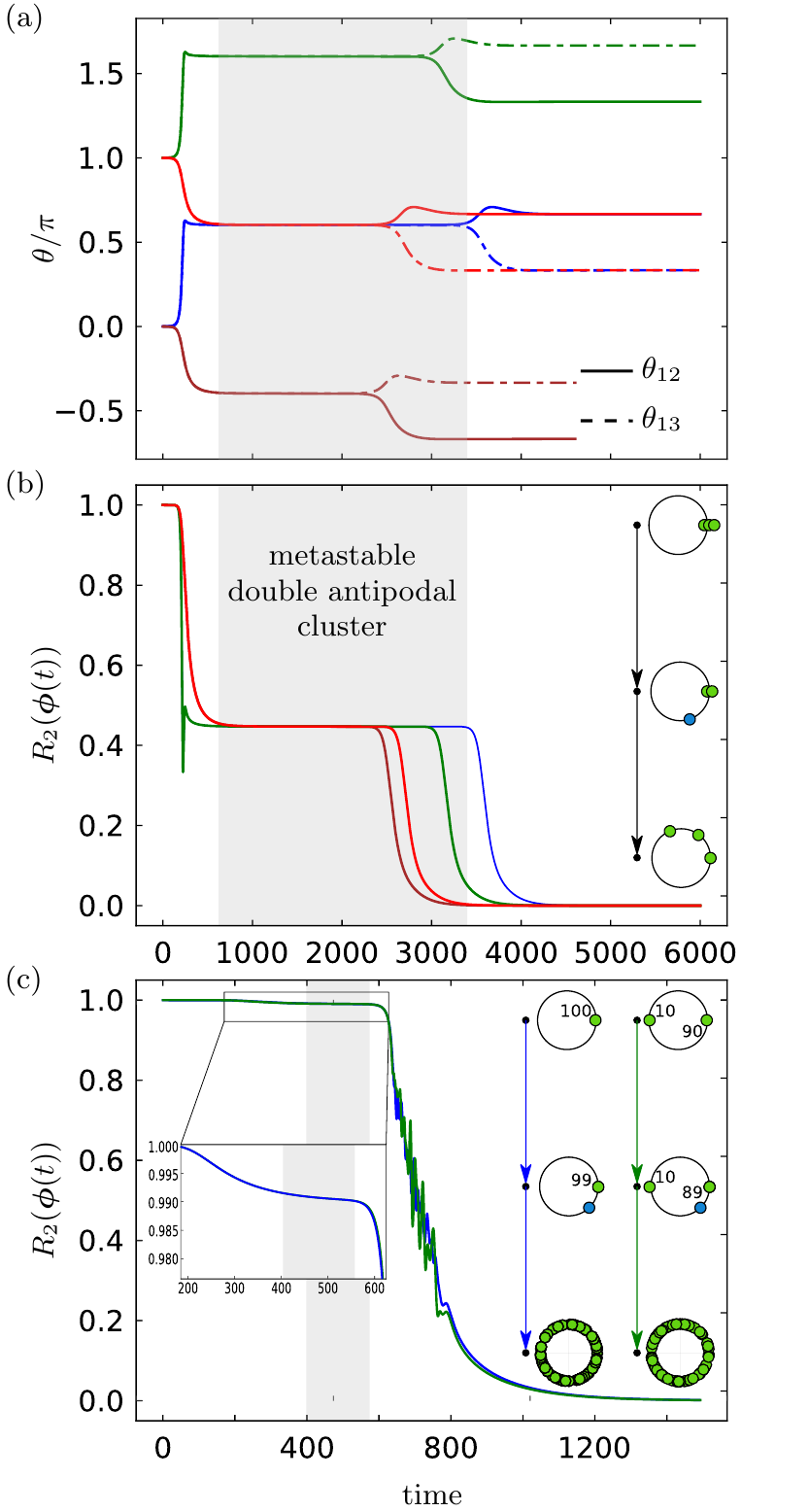}
	\caption{\label{fig:DoubleAntipodal} Heteroclinic orbits between several steady
		states in a system of $3$ and $100$ adaptively coupled phase oscillators. (a) The time series for the relative phases $\theta_{12}$ (solid lines) and $\theta_{13}$ (dashed lines) for $N=3$ are shown. Lines with the same colour correspond to the same trajectories. Panel (b,c) show time series for the second moment order parameter $R_2(\bm{\phi}(t))$ as well as a schematic illustration of the observed heteroclinic connections (right) for (b) $N=3$ and (c) $N=100$. Parameter values: $\epsilon=0.01$, $\alpha=0.4\pi$, and $\beta=-0.15\pi$.}
\end{figure}

In Fig.~\ref{fig:DoubleAntipodal}(a) we present trajectories which initially start close to antipodal clusters. The trajectories in phase space are represented by the relative coordinates $\theta_{12}=\phi_1-\phi_2$ and $\theta_{13}=\phi_1 - \phi_3$. In particular, the two configurations with $(\theta_{12}=0,\theta_{13}=0)$ and $(\theta_{12}=0,\theta_{13}=\pi)$ are considered. The coupling weights are initialized according to Eq.~\eqref{eq:1Clcouplings}. With the given parameters, the unstable manifold of the antipodal state is one-dimensional which can be determined via Proposition~\ref{prop:LinerizedOneCluster_RW_Spectrum}. For the numerical simulation, we perturb the antipodal state in such a way that two distinct orbits close to the unstable manifold are visible. For both configuration the two orbits are displayed in Fig.~\ref{fig:DoubleAntipodal}(a). It can be observed that after leaving the antipodal state the trajectories approach the
double antipodal states before leaving it towards the direction of a splay state. With this we numerically find orbits close to 
"heteroclinic", which connect antipodal, double antipodal, and splay
clusters, see schematic picture on the right in Fig.~\ref{fig:DoubleAntipodal}(b). 
The phase differences $\theta_{12}$ and $\theta_{13}$ at the double antipodal state agree with the solution $\psi$ of Eq.~(\ref{eq:DoubleAntipodalCond}) or $\psi+\pi$.
Figure~\ref{fig:DoubleAntipodal}(b)
further justifies our statements on the heteroclinic contours. Here,
we see the time series for the second moment order parameter
for all trajectories in Fig.~\ref{fig:DoubleAntipodal}(b). It can
be seen that in all cases we start at an antipodal cluster ($R_{2}(\bm{\phi})=1$)
from which the double antipodal state ($R_{2}(\bm{\phi})\approx0.447$, theoretical)
is quickly approached. The trajectories stay close to the double antipodal
cluster for  approximately $2000$ time units (shaded area) before
leaving the invariant set towards the splay state ($R_{2}(\bm{\phi})=0$).

As a second example we analyze a system of $N=100$ adaptively coupled phase oscillators. Here, we choose two particular antipodal states as initial condition and add a small perturbation to both. One of the states is chosen as an in-phase synchronous cluster. In both cases, the couplings weights are initialized in accordance with Eq~\eqref{eq:1Clcouplings}. For Figure~\ref{fig:DoubleAntipodal}(c) we depict the trajectories which show a clear heteroclinic contour between antipodal, double antipodal, and splay state as in the example of three phase oscillators. We illustrate the heteroclinic connections and present $R_2(t)$ for the corresponding trajectories in Figure~\ref{fig:DoubleAntipodal}(c). Here, the zoomed view clearly shows that the trajectories for both initial conditions starting at an antipodal cluster ($R_{2}(\bm{\phi})=1$) again first approaches an double antipodal state ($R_{2}(\bm{\phi})\approx0.990$, theoretical) before leaving it towards a splay state ($R_{2}(\bm{\phi})=0$). More precisely, each trajectory comes close to a particular double antipodal state for which only one oscillator has a phase in $\{\psi,\psi+\pi\}$. Remarkably, these states, also known as solitary states, have been found in a range of other systems of coupled oscillators, as well~\cite{Maistrenko2014}. With Proposition~\ref{prop:stab_double_antipodal} in the Appendix, one can show that this double antipodal states have a stable manifold with co-dimension one which thus divides the phase space. Next to this fact, numerical evidence for the existence of heteroclinic connections between antipodal and double antipodal states as well as between double antipodal and the family of splay states is provided in Figure~\ref{fig:DoubleAntipodal}(c). With this, double antipodal states play an important role for the organization of the dynamics in system~\eqref{eq:PhiDGL_general}--\eqref{eq:KappaDGL_general}. 

\section{Emergence of multi-cluster states and the role of the time-separation parameter}\label{sec:RoleSlowAdap}
In this section we show the importance of the time-separation parameter $\epsilon$ for the appearance of multi-cluster states. In particular, we obtain the critical value $\epsilon_c$ above which the multi-cluster states cease to exist. In order to shed light on the nature of multi-cluster states which
are built out of one-cluster states, we review the following analytic
result for two-cluster states of splay type which was explicitly derived in Ref.~\onlinecite{BER19}. Suppose we have two groups
of splay cluster states, \textit{i.e.}, $R_{2}(\mathbf{a}_{\mu})=0$
for $\mu=1,2$. Then their combination leads to the following two-cluster
solution of the system~(\ref{eq:PhiDGL_general})--(\ref{eq:KappaDGL_general})
\begin{align*}
	\phi_{i,1}(t) & =\Omega_{1}t+a_{i,1},\quad i=1,\dots,N_{1}\\
	\phi_{i,2}(t) & =\Omega_{2}t+a_{i,2},\quad i=1,\dots,N_{2}\\
	\kappa_{ij,\mu\mu}(t) & =-\sin(a_{i,\mu}-a_{j,\mu}+\beta),\quad\mu=1,2\\
	\kappa_{ij,\mu\nu}(t) & =-\rho_{\mu\nu}\sin(\Delta\Omega_{\mu\nu}t+a_{i,\mu}-a_{j,\nu}+\beta-\psi_{\mu\nu}),
\end{align*}
where $\mu\ne\nu$, $\Delta\Omega_{\mu\nu}:=\Omega_{\mu}-\Omega_{\nu}$,
$\psi_{\mu\nu}:=\arctan(\Delta\Omega_{\mu\nu}/\epsilon),$
\[
\rho_{\mu\nu}:=\left(1+\left(\Delta\Omega_{\mu\nu}/\epsilon\right)^{2}\right)^{-\frac{1}{2}},
\]

\begin{align*}
	\Omega_{\mu}=\frac{1}{2}\left(n_{\mu}\cos(\alpha-\beta)+\rho_{\mu\nu}(1-n_{\mu})\cos(\alpha-\beta+\psi_{\mu\nu})\right),
\end{align*}

\begin{multline}
	\left(\Delta\Omega_{12}\right)_{1,2}=\frac{1}{2}\left(n_{1}-\frac{1}{2}\right)\cos(\alpha-\beta)\\
	\pm\frac{1}{2}\sqrt{\left(n_{1}-\frac{1}{2}\right)^{2}\cos^{2}(\alpha-\beta)-2\epsilon(2\epsilon+\sin(\alpha-\beta))},\label{eq:2Cluster_OmegaDiff}
\end{multline}
and $n_{1}:=N_{1}/N$.  

It is quite remarkable that in case of splay-type clusters the multi-cluster
solution can be explicitly given. For a proof we refer the reader to
Ref.~\onlinecite{BER19}.

\begin{figure*}
	\centering \includegraphics{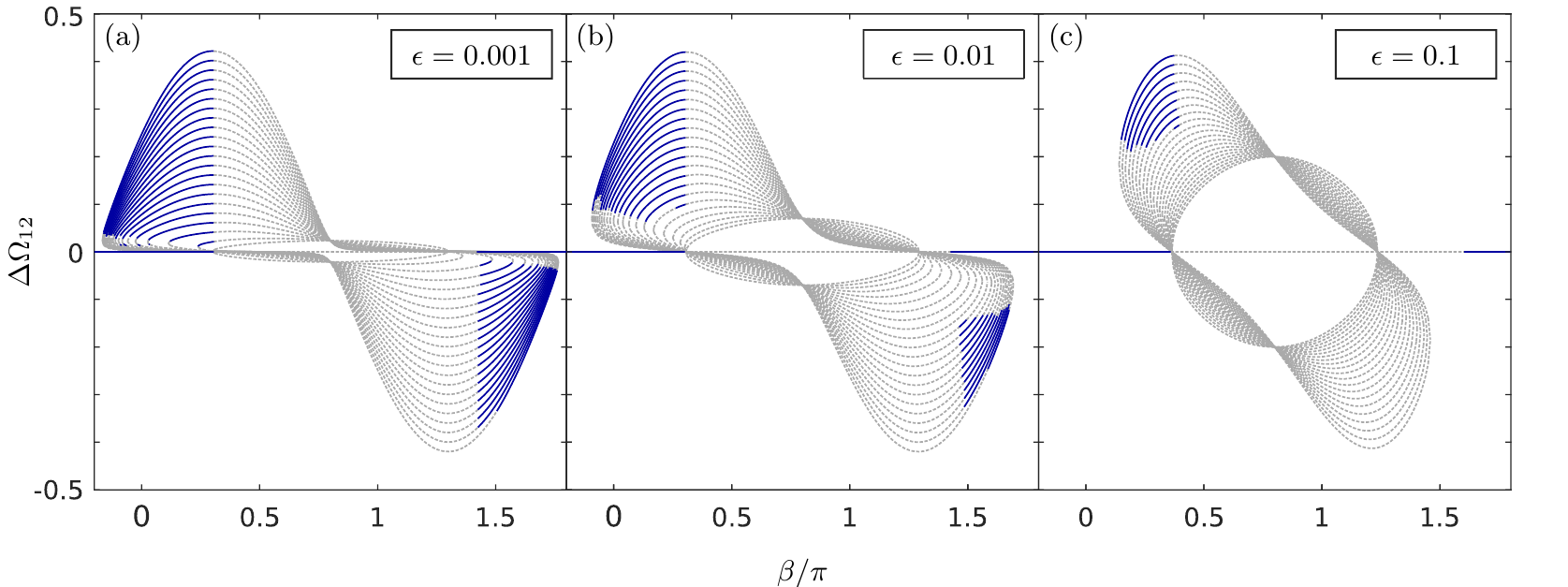}
	\caption{\label{fig:2Cluster_BifDia} All possible  one- and two-cluster solutions
		of splay type of system (\ref{eq:PhiDGL_general})--(\ref{eq:KappaDGL_general}).
		For fixed relative size of the first cluster $n_1 $, the frequency differences $\Delta\Omega_{12}(\beta)$ are displayed as a function of the system parameter $\beta$ corresponding
		to Eq.~\eqref{eq:2Cluster_OmegaDiff}. The dashed lines (gray) indicate
		unstable solutions while the solid lines (blue) indicate stable
		solutions. Parameter values $N=50$ and $\alpha=0.3\pi$ are fixed
		for all panels.}
\end{figure*}

With this, we can study directly the role
of several parameters for the existence of multi-cluster states. In
Figure~\ref{fig:2Cluster_BifDia}, solutions for Eq.~\eqref{eq:2Cluster_OmegaDiff}
are presented depending on the parameter $\beta$. The  number of oscillators
in the system is chosen as $N=50$. Each line $\Delta\Omega_{12}(\beta)$ in Fig.~\ref{fig:2Cluster_BifDia}
represents a frequency difference of two clusters for which the two-cluster
state of splay type exists with fixed relative number $n_{1}$ of oscillators in the first cluster.
Note that the number of possible two-cluster states increases proportionally
to the total number of oscillators $N$. Different panels show solutions
for different values of $\epsilon$. We note that the existence of those
two-cluster states depends only on the difference of $\gamma:=\alpha-\beta$,
see Eq.~\eqref{eq:2Cluster_OmegaDiff}. The necessary condition for
the existence of a two-cluster state reads  
\begin{align}
	\left(n_{1}-\frac{1}{2}\right)^{2}\cos^{2}\gamma>2\epsilon(2\epsilon+\sin\gamma).\label{eq:2Cl_splay_ExCond}
\end{align}

From Eq.~\eqref{eq:2Cl_splay_ExCond} we immediately see that the
value  of the time separation parameter $\varepsilon$ in system~\eqref{eq:PhiDGL_general}--\eqref{eq:KappaDGL_general}
is important for the existence of the multi-cluster states. This dependence
is in contrast to the findings for one-cluster states. First of all,
note that the left hand side of condition~\eqref{eq:2Cl_splay_ExCond}
is positive for  $\gamma\ne\pm\pi/2$. Hence, for all parameters,
there is a critical value $\epsilon_{c}$ such that there exists no
two-cluster state for $\epsilon>\epsilon_{c}$. Explicitly, we have
\begin{align}
	\epsilon_{c}=-\frac{1}{4}\sin\gamma+\frac{1}{2}\sqrt{\frac{1}{4}\sin^{2}\gamma+\left(n_{1}-\frac{1}{2}\right)^{2}\cos^{2}\gamma},\label{eq:Eps_critical}
\end{align}
which is illustrated in Fig.~\ref{fig:2Cluster_Ex_EpsCrit}. The
figure shows the critical value  $\epsilon_{c}$ depending on the
parameter $\gamma$ for different values of $n_{1}$. The function
possesses a global maximum with $\epsilon_{c}=0.5$. This means that
there is a particular requirement on the time separation in order
to have two-cluster states of splay type. Indeed, the adaptation of
the network has to be at most half as fast as the dynamics of the
oscillatory system. 

Further let us remark that the two-cluster state with equally sized
clusters $n_{1}=0.5$  exists only for $\alpha-\beta\in(\pi,2\pi)$,
\textit{i.e.}, $\epsilon_{c}=0$ for all $\alpha-\beta\in[0,\pi]$.

\begin{figure}
	\centering 
	\includegraphics{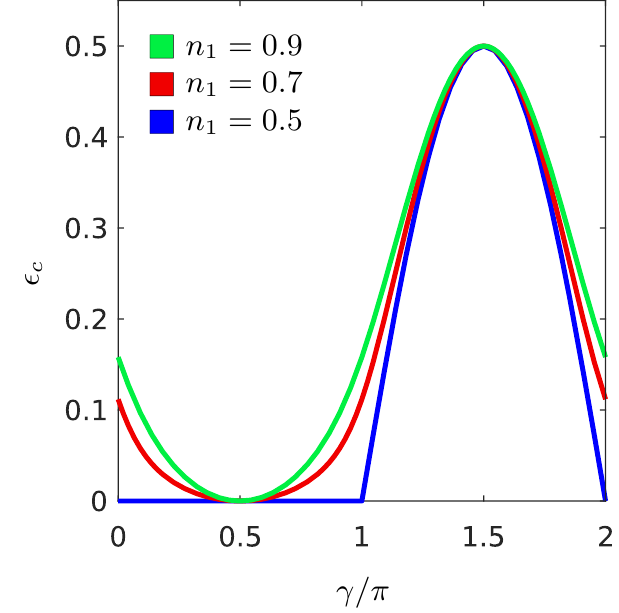}
	\caption{\label{fig:2Cluster_Ex_EpsCrit} For the case of two-cluster states
		of splay type, the critical value $\epsilon_{c}$ of time-separation
		parameter $\epsilon$ is plotted as a function of $\gamma=\alpha-\beta$
		for different cluster sizes $n_{1}=N_{1}/N$. The function is given
		explicitly by Eq.~\eqref{eq:Eps_critical}. }
\end{figure}

In Sec.~\ref{sec:Hcluster} we discussed that the combination of
one-cluster states to a multi-cluster state can result in modulated
dynamics of the oscillators additional to the linear growth. In fact,
this additional temporal behavior is due to the interaction of the
clusters. As we see in Fig.~\ref{fig:MC_states_types}(a), oscillators
interacting with a splay-type cluster will not be forced to perform
additional dynamics. This is the reason why we are able to derive
a closed analytic expression for multi-cluster states of splay type. It
is possible to determine the frequencies explicitly as in Eq.~\eqref{eq:2Cluster_OmegaDiff}.
Therefore, here the interaction between cluster causes only changes
in the collective frequencies which are small whenever $\epsilon$
is small.

In contrast to splay clusters, the interaction with antipodal cluster
leads to bounded modulation of the oscillator dynamics besides the constant-frequency motion. The modulations scale with $\epsilon$ and hence depend on the time-separation parameter. More rigorous results can be found in Ref.~\onlinecite{BER19}. In case of a mixed type two-cluster state, both interaction phenomena are present. The oscillators in the antipodal cluster interact with the splay cluster leading to no additional modulation. On the contrary, the phase oscillators of the splay cluster get additional modulation via the interaction with the antipodal cluster.

The one-cluster states of any size apparently serve as building blocks
for multi-cluster states.  However, not all possible multi-cluster
states are stable even though the building blocks are. The next section is devoted to this question of
stability.
\section{Stability of multi-cluster states}\label{sec:MClstability}
As mentioned above, the stability of one-cluster states is important
for the stability of multi-cluster states. For two weakly coupled
clusters, the stability of one-clusters serves as a necessary condition
for the stability of the two-cluster state. In figure~\ref{fig:2Cluster_BifDia}
the possible two-clusters of splay type are plotted. In addition,
for each of these solutions the stability is analysed numerically.
For this, we initialize the system~\eqref{eq:PhiDGL_general}\textendash \eqref{eq:KappaDGL_general}
on the corresponding two-cluster state and run the simulation for
$t=10000$ time units. After the simulation, we compare the initial
condition with the final state in order to determine the stability.
For each parameter value $\beta$ we colour the line blue (solid
line) whenever the two-cluster state is stable. Otherwise, the line
is gray (dashed line). An additional line at $\Delta\Omega_{12}=0$
is plotted corresponding to the one-cluster solution. The stability
of the one-cluster solution is determined analytically as in Fig.~\ref{fig:StabRegions}.
We notice that for $\epsilon=0.001$ a stable two-cluster state exists
for almost every relative cluster size $n_{1}$, while this is not
true for $\epsilon=0.01$ and even more so for $\epsilon=0.1$.

Another observation from Figs.~\ref{fig:2Cluster_BifDia}(a),(b) is
that the possible $\beta$-values where the two-cluster states can
be stable mainly correspond to the $\beta$ values where the one-cluster
state is stable. This is true for small values of $\epsilon$, however,
a careful inspection of Fig.~\ref{fig:2Cluster_BifDia}(c) for the
case of larger $\epsilon$, here $\epsilon=0.1$, shows that some two-cluster states
appear to be stable for a parameter region where the corresponding one-cluster
state is unstable. This can be explained as follows. According to
$\eqref{eq:PhiDGL_general}$, in the case of one-cluster states, the
inter-cluster interactions are summed over all $N$ oscillators of
the whole system. Additionally, the interactions are scaled with the
factor $1/N$. Therefore, the total interaction scales with $1$.
For two-cluster states, the inter-cluster interactions for each individual
cluster are only a sum over the $N_{\mu}$ ($\mu=1,2$) oscillators
whereas the scaling remains $1/N$. Hence, the total inter-cluster
interaction scales with $n_{\mu}=N_{\mu}/N$, the relative size of
the cluster. Therefore, the effective oscillatory system, when neglecting
the interaction to the other cluster, reads 
\begin{align*}
	\frac{\mathrm{d}\phi_{i,\mu}}{\mathrm{d}t} & =-\frac{n_{\mu}}{N_{\mu}}\sum_{j=1}^{N_{\mu}}\kappa_{ij,\mu\mu}\sin(\phi_{i,\mu}-\phi_{j,\mu}+\alpha),\\
	\frac{\mathrm{d}\kappa_{ij,\mu\mu}}{\mathrm{d}t} & =-\epsilon\left(\kappa_{ij,\mu\mu}+\sin(\phi_{i,\mu}-\phi_{j,\mu}+\beta)\right).
\end{align*}
This system is equivalent to ~\eqref{eq:PhiDGL_general}\textendash \eqref{eq:KappaDGL_general}
with $N_{\mu}$ oscillators by rescaling $\epsilon\mapsto\epsilon/n_{\mu}$.
Thus, the stability of the inter-cluster system has to be evaluated
with respect to the rescaled effective parameter $\epsilon_{\text{eff}}:=\epsilon/n_{\mu}$.
Since $n_{\mu}<1$ for $\mu=1,2$, we have $\epsilon_{\text{eff}}>\epsilon$.
As we have discussed, the stability for the one-cluster changes with
increasing $\epsilon$. With this, the influence of the cluster size
as well as the slight boundary shift in the regions of stability,
see Fig.~\ref{fig:2Cluster_BifDia}, can be explained.

Finally, we note why the equally-sized splay-clusters are not found to be stable. Indeed, from Eq.~\eqref{eq:2Cluster_OmegaDiff} we know that $2\epsilon+\sin(\alpha-\beta)<0$ is a necessary condition to have
such equally-sized ($n_{1}=1/2$) clusters. However, any one-cluster
splay state is unstable for $2\epsilon+\sin(\alpha-\beta)<0$ by Proposition~\ref{prop:NecCondSplay} in Appendix~\ref{sec:app_StabOneCluster}.
\section{Conclusion}\label{sec:conclusion}
In summary, we have studied a paradigmatic model of adaptively coupled phase oscillators. It is well known that various models of weakly coupled oscillatory systems can be reduced to coupled phase oscillators. Our study has revealed the impact of synaptic plasticity upon the collective dynamics of oscillatory systems. For this, we have implemented a simplified model which is able to describe the slow adaptive change of the network depending on the oscillatory states. The slow adaptation is controlled by a time-scale separation parameter.

We have described the appearance of several different frequency cluster states. Starting from random initial conditions, our numerical simulations show two different types of states. These are the splay and the antipodal type multi-cluster states. A third mixed type multi-cluster state is found by using mathematical methods described in detail in~\cite{BER19}. For all these states the collective motion of oscillators, the shape of the network, and the interaction between the frequency clusters is presented in detail. It turns out that the oscillators are able to form groups of strongly connected units. The interaction between the groups is weak compared to the interaction within the groups. The analysis of multi-cluster states reveals the building blocks for these states.

In particular, the following three types of relative equilibria form building blocks for multi-cluster states: splay, antipodal, and double antipodal. In order to understand the stability of the frequency cluster states, we perform a linear stability analysis for the relative equilibria. The stability of these states is rigorously described, and the impact of all parameters is shown. We prove that the double antipodal states are unstable in the whole parameter range. They appear to be saddle-points in the phase space which therefore cannot be building blocks for higher multi-cluster states. While the time-scale separation has no influence upon their existence, it plays an important role for the stability the relative equilibria. The regions of stability in parameter space are presented for different choices of the time-scale separation parameter. The singular limit ($\epsilon \to 0$) and the limit of instantaneous adaptation are analyzed. The latter shows that the stability region of the splay and the antipodal states divide the whole space into two equally sized regions without intersection. Instantaneous adaption cancels multistability of these states. The consideration of the singular limit shows that it differs from the case of no adaptation. Therefore, even for very slow adaptation, the oscillatory dynamics alone is not sufficient to describe the stability of the system.

Subsequently, the role of double antipodal states is discussed. We find that in a system of $3$ oscillators these states are transient states in a small heteroclinic network between antipodal and splay states. They appear to be metastable, i.e. observable for a relative long time and therefore are physically important transient states. Moreover an additional analysis for an ensemble of $100$ phase oscillators has revealed the importance of the double antipodal states for the global dynamics of the whole system.

For the splay clusters we analytically show the existence of two-cluster states. Remarkably, while the existence of the one-cluster states does not depend on the time-scale separation parameter, the multi-cluster states crucially depend on the time-scale separation. In fact, we provide an analysis showing that there exists a critical value for the time-scale separation. Moreover, we show that in the case of two-cluster states of splay type the adaption of the coupling weights must be at most half as fast as the dynamics of the oscillators. This fact is of crucial importance for comparing dynamical scenarios induced by short-term or long-term plasticity~\cite{Frohlich2016}.

The stability of two-cluster states is analyzed numerically and presented for different values of the time-scale separation parameter. By assuming weakly interacting clusters, we describe the stability of the two-cluster with the help of the analysis of one-cluster states. The simulations show that there are no stable two-cluster states with clusters of the same size. We provide an argument to understand this property of the system.

\begin{acknowledgments}
	The authors acknowledge financial support by the Deutsche Forschungsgemeinschaft
	(DFG, German Research Foundation) - Projects 411803875 (S.Y) and 308748074 (R.B.).
	V.N. and D.K. acknowledge the financial support of the Russian Foundation for Basic Research (Project Nos. 17-02-00874, 18-29-10040 and 18-02-00406).
	
\end{acknowledgments}

%

\appendix
\section{Stability of one-cluster states\label{sec:app_StabOneCluster}}
In order to study the local stability of one-cluster solutions described in Sec.~\ref{sec:blocks}, we linearise the system of differential equations (\ref{eq:PhiDGL_general})--(\ref{eq:KappaDGL_general})
around the phase locked states described by $\phi_i=\Omega t + a_i$  and $\kappa_{ij}=-\sin(a_i-a_j+\beta)$. We obtain
the following linearised system 
\begin{multline}\label{eq:Linearized_OneCl_phi}
	\frac{d}{dt}\delta\phi_{i}=\frac{1}{2N}\sum_{j=1}^{N}\sin(\beta-\alpha)\left(\delta\phi_{i}-\delta\phi_{j}\right)\\+\frac{1}{2N}\sum_{j=1}^{N}\cos(2(a_{i}-a_{j})+\alpha+\beta)\left(\delta\phi_{i}-\delta\phi_{j}\right)\\
	-\frac{1}{N}\sum_{j=1}^{N}\sin(a_{i}-a_{j}+\alpha)\delta\kappa_{ij},
\end{multline}
and 
\begin{align}\label{eq:Linearized_OneCl_kappa}
	\frac{d}{dt}\delta\kappa_{ij} =-\epsilon\left(\delta\kappa_{ij}+\cos(a_{i}-a_{j}+\beta)\left(\delta\phi_{i}-\delta\phi_{j}\right)\right),
\end{align}
Note that this set of equations can be brought into the following block form
\begin{align}\label{eq:LinearisationBlockForm}
	\frac{\mathrm{d}}{\mathrm{d}t}
	\begin{pmatrix}
		\delta\bm{\phi}\\
		\delta\kappa
	\end{pmatrix}
	=
	\begin{pmatrix}
		A & B\\
		C & -\epsilon\mathbb{I}_{N^2}
	\end{pmatrix} 
	\begin{pmatrix}
		\delta\bm{\phi}\\
		\delta\kappa
	\end{pmatrix}
\end{align}
where $\left(\delta\bm{\phi}\right)^{T}=\left(\delta\phi_{1},\dots,\delta\phi_{N}\right)$, $\left(\delta\kappa\right)^{T}=\left(\delta\kappa_{11},\dots,\delta\kappa_{1N},\delta\kappa_{21},\dots,\delta\kappa_{NN}\right)$, $B=\begin{pmatrix}B_1 & \cdots & B_N
\end{pmatrix}$, $C=\begin{pmatrix}C_1\\ \vdots \\C_N
& \\
\end{pmatrix}$, and $A$, $B_n$, $C_n$ are $N\times N$ matrices with $n=1,\dots,N$. The elements of the block matrices read
\begin{align*}
a_{ij} & =\begin{cases}
\!\begin{aligned}
-\frac{1}{2}\sin(\alpha-\beta)-\frac{1}{N}\sin(\beta)\cos(\alpha) \\
+\frac{1}{2N}\sum_{k=1}^{N}\sin(2(a_{i}-a_{k})+\alpha+\beta),
\end{aligned} 
 & i=j\\
\frac{1}{2N}\left(\sin(\alpha-\beta)-\sin(2(a_{i}-a_{j})+\alpha+\beta)\right), & i\ne j
\end{cases}\\
b_{ij;n} & =\begin{cases}
-\frac{1}{N}\sin(a_{n}-a_{j}+\alpha), & i=n\\
0, & \text{otherwise}
\end{cases}\\
c_{ij;n} & =\begin{cases}
0, & j=n,i=j\\
-\epsilon\cos(a_{n}-a_{i}+\beta), & j=n,i\ne j\\
\epsilon\cos(a_{n}-a_{i}+\beta), & j\ne n,i=j\\
0, & \text{otherwise}
\end{cases}.
\end{align*}
Throughout this appendix we will make use of Schur's complement~\cite{BOY04} in order to simplify characteristic equations. In particular, any $m\times m$ matrix $M$ in the $2\times2$ block form can be written as
\begin{multline}\label{eq:SchurComplement}
M =\begin{pmatrix}A & B\\
C & D
\end{pmatrix} \\
=\begin{pmatrix}\mathbb{I}_{p} & BD\\
0 & \mathbb{I}_{q}
\end{pmatrix}\begin{pmatrix}A-BD^{-1}C & 0\\
0 & D
\end{pmatrix}\begin{pmatrix}\mathbb{I}_{p} & 0\\
D^{-1}C & \mathbb{I}_{q}
\end{pmatrix}
\end{multline}
where $A$ is a $p\times p$ matrix and $D$ is an invertible $q\times q$
matrix. The matrix $A-BD^{-1}C$ is called Schur's
complement. A simple formula for the determinant of $M$ can be derived
with this decomposition in~\eqref{eq:SchurComplement}
\begin{align*}
\det(M) & =\det(A-BD^{-1}C)\cdot\det(D).
\end{align*}
This result is important for the subsequent stability analysis. Note that in the following an asterisk indicates the complex conjugate.
\begin{prop} \label{prop:LinerizedOneCluster_RW_Spectrum}
	Suppose we have $\mathbf{a}_{k}=(0,\frac{2\pi}{N}k,\dots,(N-1)\frac{2\pi}{N}k)^{T}$
	and the characteristic equation of the linear system (\ref{eq:Linearized_OneCl_phi})--(\ref{eq:Linearized_OneCl_kappa})
	then the set of eigenvalues $L$ are as follows.
	\begin{enumerate}
		\item (in-phase and anti-phase synchrony) If $k=0$ or $k=N/2$, 
		\[
		L=\left\{ \left(0\right)_{\text{1}},\left(-\epsilon\right)_{(N-1)N+1},\left(\lambda_{1}\right)_{N-1},\left(\lambda_{2}\right)_{N-1}\right\} 
		\]
		where $\lambda_{1}$ and $\lambda_{2}$ solve $\lambda^{2}+\left(\epsilon-\cos(\alpha)\sin(\beta)\right)\lambda-\epsilon\sin(\alpha+\beta)=0$.
		\item (incoherent rotating-wave) If $k\ne0,N/2,N/4,3N/4,$ the spectrum
		is 
		\begin{multline*}
			L =\left\{ \left(0\right)_{N-2},\left(-\epsilon\right)_{(N-1)N+1},\left(-\frac{\sin(\alpha-\beta)}{2}-\epsilon\right)_{N-3},\right.\\
			\left.\left(\vartheta_{1}\right)_{1},\left(\vartheta_{2}\right)_{1},\left(\vartheta_{1}^{*}\right)_{1},\left(\vartheta_{2}^{*}\right)_{1}\right\} 
		\end{multline*}
		where $\vartheta_{1}$ and $\vartheta_{2}$ solve $\vartheta^{2}+\left(\epsilon+\frac{1}{2}\sin(\alpha-\beta)-\frac{1}{4}\mathrm{i}e^{\mathrm{i}(\alpha+\beta)}\right)\vartheta-\frac{\epsilon}{2}\mathrm{i}e^{\mathrm{i}(\alpha+\beta)}=0$.
		\item (4-rotating-wave state) If $k=N/4,3N/4$, the spectrum is 
		\begin{multline*}
			L =\left\{ \left(0\right)_{N-1},\left(-\epsilon\right)_{(N-1)N+1},\right.\\
			\left. \left(-\frac{\sin(\alpha-\beta)}{2}-\epsilon\right)_{N-2},\left(\lambda_{1}\right)_{1},\left(\lambda_{2}\right)_{1}\right\} 
		\end{multline*}
		where $\lambda_{1}$ and $\lambda_{2}$ solve $\lambda^{2}+\left(\epsilon+\sin(\alpha)\cos(\beta)\right)\lambda+\epsilon\sin(\alpha+\beta)=0.$
	\end{enumerate}
	Here, the multiplicities for each eigenvalue are given as subscripts.
\end{prop}
\begin{proof}
	The proof can be found in Ref~\onlinecite{BER19}.
\end{proof}
So far, we have found the Lyapunov coefficients for the rotating-wave states. The following two Lemmata are needed to describe the stability of antipodal, $4$-phase-cluster, and double antipodal states as well. 
\begin{lem}\label{lem:BlockCirculantMatrix}
	Suppose $M$ is a block square matrix
	of the form
	\begin{align*}
		M & =\left(\begin{array}{cc}
			A & m_{1}\hat{1}_{p,q}\\
			m_{2}\hat{1}_{q,p} & B
		\end{array}\right)
	\end{align*}
	where $A$ is a circulant $p\times p$ matrix, $B$ is a circulant $q\times q$ matrix, $\hat{1}_{p,q}$ is $p\times q$ where all entries are $1$ and $m_{1},m_{2}\in\mathbb{R}$. Then the eigenvector-eigenvalue	pairs are given by 
		\begin{align}
			&\left(\lambda_{k}^{0},\dots,\lambda_{k}^{p-1},0,\dots,0\right)^{T},  \,\mu_{k}=\sum_{l=0}^{p-1}a_{1\left(1+l\right)}\lambda_{k}^{l} \label{eq:EVEQpair_1}\\
			&\left(0,\dots,0,\rho_{l}^{0},\dots,\rho_{l}^{q-1}\right)^{T},  \,\nu_{l}=\sum_{l=0}^{q-1}b_{1\left(1+l\right)}\rho_{l}^{l} \label{eq:EVEQpair_2}\\
			&\left(1,\dots,1,a_{1},\dots a_{1}\right)^{T}, \,\bar{\mu}=\mu_{0}+m_{1}qa_{1} \label{eq:EVEQpair_3}\\
			&\left(1,\dots,1,a_{2},\dots,a_{2}\right)^{T}, \,\bar{\nu}=\mu_{0}+m_{1}qa_{2} \label{eq:EVEQpair_4}
		\end{align}
		with $\lambda_{k}=e^{\mathrm{i}\frac{2\pi}{p}k}$ and $\rho_{l}=e^{\mathrm{i}\frac{2\pi}{q}l}$ for $k=1,\dots,p-1$ and $l=1,\dots,q-1$, respectively, and $a_{1}$ and $a_{2}$ solve the equation 
		\begin{align*}
			a^{2}+\frac{\mu_{0}-\nu_{0}}{m_{1}q}a-\frac{m_{2}p}{m_{1}q} & =0
		\end{align*}
		with $\mu_{0}=\sum_{j=1}^{p}a_{1j}$ and $\nu_{0}=\sum_{j=1}^{q}b_{1j}$.
\end{lem}
\begin{proof}
	We can prove the Lemma by direct calculation and find
	\begin{multline*}
		M\left(\lambda_{k}^{0},\dots,\lambda_{k}^{p-1},0,\dots,0\right)^{T}\\=\left(\begin{array}{c}
			A\left(\lambda_{k}^{0},\dots,\lambda_{k}^{p-1}\right)^{T}
			m_{2}\sum_{l=0}^{p-1}\lambda_{k}^{l}
		\end{array}\right)\\
		=\mu_{k}\left(\lambda_{k}^{0},\dots,\lambda_{k}^{p-1},0,\dots,0\right)^{T}.
	\end{multline*}
	Here, we use that $A$ is a circulant matrix and that $\sum_{l=0}^{p-1}\lambda_{k}^{l}=0$
	for all $k=1,\dots,p-1$. Analogous arguments hold for \eqref{eq:EVEQpair_2}. The last two eigenvector-eigenvalue pairs \eqref{eq:EVEQpair_3}--\eqref{eq:EVEQpair_4} can be obtained by
	\begin{multline*}
		M \left(1,\dots,1,a,\dots,a\right)^{T}\\=\left(\begin{array}{c}
			A\left(1,\dots,1\right)^{T}+m_{1}qa\left(1,\dots,1\right)^{T}\\
			\frac{m_{2}p}{a}\left(1,\dots,1\right)^{T}+aB\left(1,\dots,1\right)^{T}
		\end{array}\right)\\
		=\left(\begin{array}{c}
			\mu_{0}+m_{1}qa\\
			\frac{m_{2}p}{a}+\nu_{0}
		\end{array}\right)\left(1,\dots,1,a,\dots,a\right)^{T}
	\end{multline*}
	which solves the eigenvalue problem if $a$ is chosen to be either
	$a_{1}$ or $a_{2}$.
\end{proof}
\begin{lem}\label{lem:PhaseLockSchur}
	Suppose we have a phase locked state with phases $a_i\in[0,2\pi)$. Then, the solution for the characteristic equations corresponding to the linearised system~\eqref{eq:Linearized_OneCl_phi}--\eqref{eq:Linearized_OneCl_kappa} are given by $\lambda=-\epsilon$ with multiplicity $N^2-N$ and by the solution of the following set of equations
	\begin{align*}
		\det\left(\left(A-\lambda\mathbb{I}_{N}\right)\left(\epsilon+\lambda\right)+BC\right) = 0.
	\end{align*}
\end{lem}
\begin{proof}
	Applying Schur's decomposition~\eqref{eq:SchurComplement} to the linearised system in the block form~\eqref{eq:LinearisationBlockForm} yields the result.
\end{proof}
\begin{prop}\label{prop:stab_double_antipodal}
	Suppose we have a state with phases $a_i\in\{0,\pi,\psi,\psi+\pi\}$ where $i=1,\dots,N$. Further set $q_1=Q_1/N$ and $q_2=Q_2/N$, where $Q_1$ and $Q_2$ denote the numbers of phases which are either $0$ or $\pi$ and $\psi$ or $\psi+\pi$, respectively. Then, the linear system (\ref{eq:Linearized_OneCl_phi})--(\ref{eq:Linearized_OneCl_kappa}) possesses the following set $L$ of eigenvalues
	\begin{multline*}
		L=\left\{ \left(0\right)_{\text{1}},\left(-\epsilon\right)_{(N-1)N+1},\left(\lambda_{1}\right)_{N_1-1},\left(\lambda_{2}\right)_{N_1-1},\right.\\
		\left. \left(\vartheta_1\right)_{N_2-1}, \left(\vartheta_2\right)_{N_2-1}, \left(\rho_1\right)_1, \left(\rho_2\right)_1\right\}
	\end{multline*}
	where $\lambda_{1}$ and $\lambda_{2}$ solve 
	\begin{multline*}
		\lambda^2+\frac{1}{2}\left(\sin(\alpha-\beta)-q_{1}\sin(\alpha+\beta)\right.\\
		\left.-q_{2}\sin(-2\psi+\alpha+\beta)+2\epsilon\right)\lambda\\
		-\epsilon q_{1}\sin(\alpha+\beta)-\epsilon q_{2}\sin(-2\psi+\alpha+\beta)=0,
	\end{multline*}
	$\vartheta_1$ and $\vartheta_2$ solve
	\begin{multline*}
		\vartheta^2+\frac{1}{2}\left(\sin(\alpha-\beta)-q_{1}\sin(2\psi+\alpha+\beta)\right.\\
		\left.-q_2\sin(\alpha+\beta)+2\epsilon\right)\vartheta\\
		-\epsilon q_{1}\sin(2\psi+\alpha+\beta)-\epsilon q_{2}\sin(\alpha+\beta)=0,
	\end{multline*}
	as well as $\rho_1$ and $\rho_2$ solve
	\begin{multline*}
		\rho^2+\frac{1}{2}\left(\sin(\alpha-\beta)-q_{1}\sin(2\psi+\alpha+\beta)\right.\\
		\left.-q_2\sin(-2\psi+\alpha+\beta)+2\epsilon\right)\rho\\
		-\epsilon q_{1}\sin(2\psi+\alpha+\beta)-\epsilon q_{2}\sin(-2\psi+\alpha+\beta)=0.
	\end{multline*}
	The multiplicities for each eigenvalue are given as subscripts.
\end{prop}
\begin{proof}
	For an arbitrary solution of the form $\phi_{i}=\Omega t+a_{i}$ we
	consider the linearised system~(\ref{eq:Linearized_OneCl_phi})--(\ref{eq:Linearized_OneCl_kappa}) in the block form~\eqref{eq:LinearisationBlockForm} and apply Lemma~\ref{lem:PhaseLockSchur}. The elements of the second term
	$D:=BC$ of Schur's complement are then
	\begin{align*}
		d_{ij} &=-\frac{\epsilon}{2N}\left(\sin(\alpha-\beta)+\sin(2(a_{i}-a_{j})+\alpha+\beta)\right)
	\end{align*}
	if $i\ne j$ and
	\begin{multline*}
		d_{ii} =\epsilon\left(\frac{1}{2}\sin(\alpha-\beta)-\frac{1}{N}\sin(\alpha)\cos(\beta)\right.\\
		\left. +\frac{1}{2N}\sum_{j=1}^{N}\sin(2(a_{i}-a_{j})+\alpha+\beta)\right).
	\end{multline*}
	Defining the matrix $M:=\left(A-\lambda\mathbb{I}_{N}\right)\left(\epsilon+\lambda\right)+D$
	we get
	\begin{align*}
		m_{ij} & =\begin{cases}
			-\lambda^2+(a_{ii}-\epsilon)\lambda+\epsilon a_{ii}+d_{ii} & i=j\\
			\lambda a_{ij}+\epsilon a_{ij}+d_{ij}. & i\ne j
		\end{cases}
	\end{align*}
	Using the assumption for the phases $a_i$, then one group of oscillators (group $I$) have $a_{i}\in\{0,\pi\}$ and the remaining
	$Q_{2}$ oscillators (group $II$) have $a_{i}=\{\psi,\psi+\pi\}$. Putting this into the definition of $m_{ij}$, we find that the whole square matrix $M$ can be written as
	\begin{align*}
		M & =\left(\begin{array}{cc}
			\overbrace{\begin{array}{cccc}
					m_{I} & \bar{m} & \cdots & \bar{m}\\
					\bar{m} & \ddots & \ddots & \vdots\\
					\vdots & \ddots & \ddots & \bar{m}\\
					\bar{m} & \cdots & \bar{m} & m_{I}
			\end{array}}^{Q_{1}\times Q_{1}} & \begin{array}{c}
				m_{1} \hat{1}_{Q_1,Q_2}
			\end{array}\\
			\begin{array}{c}
				m_{2} \hat{1}_{Q_2,Q_1}
			\end{array} & \underbrace{\begin{array}{cccc}
					m_{II} & \bar{m} & \cdots & \bar{m}\\
					\bar{m} & \ddots & \ddots & \vdots\\
					\vdots & \ddots & \ddots & \bar{m}\\
					\bar{m} & \cdots & \bar{m} & m_{II}
			\end{array}}_{Q_{2}\times Q_{2}}
		\end{array}\right)
	\end{align*}
	where $m_1$, $\bar{m}$, $m_I$, and $m_{II}$ are real values which depend on all the system parameters $\alpha$, $\beta$, $\epsilon$ and additionally on $\psi$ and $\lambda$. Note that all diagonal blocks are circulant matrices. The determinant is invariant under basis transformations which is why we diagonalize
	the matrix $M$ and therewith derive equations for the values $\lambda$.
	In order to do so, we look for the eigenvalues of $M$ determined
	by the characteristic equation
	\begin{align*}
		\det\left(M-\mu\mathbb{I}_{N}\right) & =0.
	\end{align*}
	Due to the structure of $M$ we can apply Lemma~\ref{lem:BlockCirculantMatrix}
	and find the following set of eigenvalues
	\begin{widetext}
	\begin{align*}
		\mu_{k} & =-\lambda^2-\frac{1}{2}(\sin(\alpha-\beta)-q_{1}\sin(\alpha+\beta)-q_{2}\sin(-2\psi+\alpha+\beta)+2\epsilon)\lambda+\epsilon q_{1}\sin(\alpha+\beta)+\epsilon q_{2}\sin(-2\psi+\alpha+\beta)
	\end{align*}
	\end{widetext}
	for $k=1,\dots,Q_{1}-1.$ Analogously, we obtain the equations for $\nu_{k}$
	($k=1,\dots,Q_{2}-1$) where $m_{I}$ is substituted with $m_{II}$.
	The two other eigenvalue are given by $\bar{\mu}=\mu_{0}+m_{1}Q_{2}a_{1}$
	and $\bar{\nu}=\mu_{0}+m_{1}Q_{2}a_{2}$, respectively, where
	\begin{align*}
		\mu_{0} & =m_{I}+(Q_{1}-1)\bar{m}
	\end{align*}
	and $a_{1,2}$ are given by 
	\begin{align*}
		a^{2}+\frac{\left(m_{I}-m_{II}\right)+\left(Q_{1}-Q_{2}\right)\bar{m}}{m_{1}Q_{2}}a-\frac{m_{2}Q_{1}}{m_{1}Q_{2}} & =0.
	\end{align*}
	Considering the row sums of $M$ we find that all agree with $-\lambda^{2}-\epsilon\lambda$
	and therefore $\bar{\mu}=-\lambda^{2}-\epsilon\lambda$. Resulting
	from this $a_{1}=-\left(\lambda^{2}+\epsilon\lambda+\mu_{0}\right)/m_{1}Q_{2}=1$.
	Hence,
	\begin{align*}
		a_{2} & =\frac{\left(m_{II}-m_{I}\right)+\left(Q_{2}-Q_{1}\right)\bar{m}}{m_{1}Q_{2}}-1
	\end{align*}
	and we find 
	\begin{align*}
		\bar{\nu} & =m_{II}+\left(Q_{2}-1\right)\bar{m}-m_{1}Q_{2}\\
		 &=-\lambda^2-\frac{1}{2}(\sin(\alpha-\beta)-q_{1}\sin(2\psi+\alpha+\beta)\\
		 &-q_2\sin(-2\psi+\alpha+\beta)+2\epsilon)\lambda\\
		&+\epsilon q_{1}\sin(2\psi+\alpha+\beta)+\epsilon q_{2}\sin(-2\psi+\alpha+\beta)
	\end{align*}
	After diagonalizing the matrix $M$ the determinant can be easily
	written as
	\begin{align*}
		\det(M) & =\bar{\mu}\cdot\mu_{1}\cdot\dots\cdot\mu_{N_{1}-1}\cdot\bar{\nu}\cdot\nu_{1}\cdot\dots\cdot\nu_{N_{2}-1}.
	\end{align*}
	Therewith, finding $\lambda's$ such that at least one of the eigenvalues
	of $M$ vanishes solves the initial eigenvalue problem.
\end{proof}
We will now sum up the results with the following corollaries
\begin{cor}\label{cor:StabAntipodal}
	The set of eigenvalues of the linearised system~ (\ref{eq:Linearized_OneCl_phi})--(\ref{eq:Linearized_OneCl_kappa}) around all antipodal states with $a_i\in\{0,\pi\}$ agrees with the set $L$ in Prop.~\ref{prop:LinerizedOneCluster_RW_Spectrum} for rotating-wave states with $k=0,N/2$.
\end{cor}
\begin{proof}
	Put $Q_2=0$ in Prop.~\ref{prop:stab_double_antipodal}, then there is only the equation for $\lambda$ left. 
\end{proof}
\begin{cor}\label{lem:Stab4phaseCl}
	The set of eigenvalues to the linearised system~ (\ref{eq:Linearized_OneCl_phi})--(\ref{eq:Linearized_OneCl_kappa}) around all $4$-phase-cluster states with $a_i\in\{0,pi/2,\pi,3\pi/2\}$ and $R_2(\bm{a})=0$ agrees with the set $L$ in Prop.~\ref{prop:LinerizedOneCluster_RW_Spectrum} for $4$-rotating-wave states.
\end{cor}
\begin{proof}
	The requirement $R_2(\bm{a})=0$ yields $Q_1=Q_2$. The statement of this proposition follows by using Prop.~\ref{prop:stab_double_antipodal}.
\end{proof}
\begin{cor}\label{cor:StabDoubleAntipodal}
	For all $\alpha$ and $\beta$ the double antipodal states are unstable.
\end{cor}
\begin{proof}
	Suppose the polynomial equation $p(x)=x^2+ax+b=0$. This equation has two negative roots if and only if $b>0$ and $a>0$ meaning that $p(0)>0$ and the vertex of the parabola is at $x<0$, respectively. In order to have stable double antipodal states these two conditions have to be met by all three equations for $\lambda$, $\vartheta$ and $\rho$ in Proposition~\ref{prop:stab_double_antipodal}. From the condition on the existence of double antipodal states~\eqref{eq:DoubleAntipodalCond} we find $q_1\sin(2\psi+\alpha+\beta)+q_2\sin(-2\psi+\alpha+\beta)=-\sin(\alpha+\beta)$. With this assumption on the quadratic equation and the latter equation, we find the following two necessary conditions for the stability of double antipodal states, (1) $q_1\sin(2\psi+\alpha+\beta)+q_2\sin(\alpha+\beta)>0$ and (2) $q_1\sin(2\psi+\alpha+\beta)+q_2\sin(\alpha+\beta)<0$. The two condition cannot be equally fulfilled.
\end{proof}
In the following, we give a necessary condition for the stability of all one-cluster states of splay type, in contrast to the result on rotating-wave states given in Prop.~\ref{prop:LinerizedOneCluster_RW_Spectrum}. In general all splay one-cluster states have the property $R_2(\bm{a})=0$ for the phase given by the vector $\bm{a}$. Therefore, the splay states form $N-2$ dimensional family of solution. Hence, around each splay states there are $N-2$ neutral variational directions $\left(\delta\bm{\phi},\delta\kappa\right)^T$ which are determined by the condition $\sum_{j=1}^N e^{\mathrm{i}2a_j}\delta\phi_j=0$. Note, $\delta\kappa_{ij}=-\cos(a_i-a_j+\beta)\left(\delta\phi_i - \phi_j\right)$ in neutral direction.
\begin{prop}\label{prop:NecCondSplay}
	Consider an asymptotically stable one-cluster state of splay type. Then, $\epsilon+\sin(\alpha-\beta)/2>0$.
\end{prop}
\begin{proof}
	Due to the block form of the linearised equation~\eqref{eq:LinearisationBlockForm} and the Schur decomposition~\eqref{eq:SchurComplement}, any eigenvalue comes with a second. We have already seen this in Lemma~\ref{lem:BlockCirculantMatrix} and Proposition~\ref{prop:stab_double_antipodal}. Variation along the neutral direction gives $N-2$ times the eigenvalue $0$. Suppose we have $\delta\bm{\phi}$ such that $\sum_{j=1}^N e^{\mathrm{i}2a_j}\delta\phi_j=0$ and $\delta\kappa_{ij}=-\cos(a_i-a_j+\beta)\left(\delta\phi_i - \phi_j\right)$. Applying Schur decomposition~\eqref{eq:SchurComplement}, we get
	\begin{widetext}
		\begin{align}\label{eq:SchurOnSplay}
			\left(M-\lambda\mathbb{I}\right)_{N^2+N} \begin{pmatrix}
			\delta\bm{\phi} \\
			\delta\kappa
			\end{pmatrix}& =\begin{pmatrix}\mathbb{I}_{N} & -(\epsilon+\lambda)B\\
			0 & \mathbb{I}_{N^2}
			\end{pmatrix}\begin{pmatrix}(A-\lambda\mathbb{I}_N)+\frac{1}{\epsilon+\lambda}BC & 0\\
			0 & -(\epsilon+\lambda)
			\end{pmatrix}
			\begin{pmatrix}
				\delta\bm{\phi} \\
				-\frac{1}{\epsilon+\lambda}C\delta\bm{\phi} + \delta\kappa
			\end{pmatrix} = 0.
		\end{align}
	\end{widetext}
	With this, we have to find $\lambda$ such that the last equality in \eqref{eq:SchurOnSplay} is fulfilled. This is equivalent to solving $\left((A-\lambda\mathbb{I}_N){(\epsilon+\lambda)}+BC\right)\delta\bm{\phi}=0$ of which in general only $N-2$ equations are linearly independent. The equivalence can be seen by multiplying $\epsilon+\lambda$ from both sides and keeping in mind that $\delta\kappa$ is already determined by $\delta\bm{\phi}$. Using the definition of $\delta\bm{\phi}$ the matrices $A$ and $BC$ can be effectively reduced in such a way that they are independent of the actual values for the phases $a_j$. In fact,
	\begin{align*}
		a_{ij} & =\begin{cases}
		-\frac{N-1}{2N}\sin(\alpha-\beta) & i=j\\
		\frac{1}{2N}\sin(\alpha-\beta), & i\ne j
		\end{cases}\\
		(bc)_{ij} & =\begin{cases}
		\epsilon\frac{N-1}{2N}\sin(\alpha-\beta), & i=j\\
		-\frac{\epsilon}{2N}\sin(\alpha-\beta). & i\ne j
		\end{cases}
	\end{align*}
	In turn, this gives $\left((A-\lambda\mathbb{I}_N){(\epsilon+\lambda)}+BC\right)$ a circulant structure which can be used to diagonalise the matrix, in analogy to Proposition~\ref{prop:stab_double_antipodal}. For circulant matrices we immediately know the eigenvalues. They are
	\begin{multline*}
		\mu_l = -\lambda^2- \left(\frac{N-1}{2N}\sin(\alpha-\beta)\right.\\
		\left.-\frac{1}{2N}\sin(\alpha-\beta)\left(\sum_{k=0}^{N-1}e^{\mathrm{i}2\pi kl/N}-1\right)+\epsilon\right)\lambda
	\end{multline*}
	with $l=0,\dots,N-1$ and $\det\left((A-\lambda\mathbb{I}_N){(\epsilon+\lambda)}+BC\right)=\mu_0(\lambda)\cdots\mu_{N-1}(\lambda)$. Remember we have in general $N-2$ independent equations. Thus, solving $\mu_l(\lambda)=0$ for $\lambda$ results in $N-2$ eigenvalues $\lambda=0$, $1$ eigenvalue $\lambda=-\epsilon$ and $N-3$ eigenvalues $\lambda=-\epsilon-\sin(\alpha-\beta)/2$. Note that for $4$-phase-cluster states, as considered in Corollary~\ref{lem:Stab4phaseCl}, the number of independent equations is $N-1$. This is due to the fact that in this case the equations for the imaginary and real part from $\sum_{j=1}^N e^{\mathrm{i}2a_j}\delta\phi_j=0$ agree.
\end{proof}
\end{document}